\newcommand{\short}[1]{#1}
\newcommand{\full}[1]{#1}
\renewcommand{\short}[1]{}

\documentclass{llncs} 

\usepackage{times}

\usepackage{xspace}

\usepackage{graphicx}
\usepackage[font=small,skip=0pt]{caption}
\usepackage{epstopdf}

\usepackage{algorithm}
\usepackage[noend]{algpseudocode}
\algnewcommand\algorithmicswitch{\textbf{switch}}
\algnewcommand\algorithmiccase{\textbf{case}}
\algdef{SE}[SWITCH]{Switch}{EndSwitch}[1]{\algorithmicswitch\ #1\ \algorithmicdo}{\algorithmicend\ \algorithmicswitch}%
\algdef{SE}[CASE]{Case}{EndCase}[1]{\algorithmiccase\ #1}{\algorithmicend\ \algorithmiccase}%
\algtext*{EndSwitch}%
\algtext*{EndCase}%

\newcommand\nat{\mathbb{N}}

\newcommand\integer{\mathbb{Z}}

\newcommand{\NBox}{\ensuremath{\mathbf{G}}}

\newcommand{\Act}{Act}

\usepackage{macrosbulling2015}

\newcommand{\rbm}{\acro{RBM}}
\newcommand{\Resources}{\resourcedomain}
\newcommand{\resourcedomain}{\ensuremath{{\mathcal{R}\!\mathit{es}}}}

\newcommand{\rhooutcome}[3]{\ensuremath{\mathit{out}(#1,#2,#3)}}

\renewcommand{\model}{\ensuremath{\mathfrak{M}}\xspace}
\newcommand{\enment}{\ensuremath{\eta}\xspace}

\newcommand{\production}{\mathsf{prod}}
\newcommand{\consumption}{\mathsf{cons}}
\newcommand{\Enments}{\ensuremath{\mathsf{En}}}

\newcommand{\brian}[1]{{\color{blue}Brian: #1}}

\renewcommand{\brian}[1]{}

\newcommand{\modelsR}{\ensuremath{\models_R}\xspace}

\newcommand{\coopdown}[2][]{\ensuremath{\coop{#2}{}_{#1}^{\downarrow}}}

\newcommand{\rbmi}{\acro{iRBM}}
\newcommand{\ral}{\ensuremath{\mathsf{RAL}}\xspace}
\renewcommand{\modelsR}{\models}

\renewcommand{\CTL}{\ensuremath{\mathsf{CTL}}\xspace}
\renewcommand{\ATL}{\ensuremath{\mathsf{ATL}}\xspace}

\newcommand{\wall}{\begin{aligned}[t] &}

\newcommand{\return}{\end{aligned}}

\newcommand{\inductioncase}[1]{\noindent\textbf{#1}}

\let\phi\varphi

\title{Model Checking Resource Bounded Systems with Shared Resources via Alternating B\"uchi Pushdown Systems}
\author{Nils Bulling\inst{1}  \and Hoang Nga Nguyen\inst{2} }
\institute{
Delft University of Technology, The Netherlands \and
School of Computer Science, University of Nottingham, UK\\
}

\begin{document}

\maketitle

\begin{abstract}
It is well known that the verification of resource-constrained multi-agent systems is undecidable in general. In many such settings, resources are private to agents. In this paper, we investigate the model checking problem for a resource logic based on Alternating-Time Temporal Logic (\ATL) with shared resources. Resources can be consumed and produced up to any amount. We show that the model checking problem is undecidable if two or more of such unbounded resources are available. Our main technical result is that in the case of a single shared resource, the problem becomes decidable. Although intuitive, the proof of decidability is non-trivial. We  reduce model checking to a problem over alternating B\"uchi pushdown systems. An  intermediate result connects to general automata-based  verification: we show that model checking Computation Tree Logic (\CTL) over (compact) alternating B\"uchi pushdown systems is decidable.
\end{abstract}

\section{Introduction and Related Work}
Research on resource-constrained multi-agent systems has become a popular topic in recent years, e.g.~\cite{BullingFarwer09rtl-clima-post,Bulling/Farwer:10a,Alechina//:10a,DellaMonica//:13a,Alechina//:14c,Bulling15ral-IJCAI}. In particular, the verification of strategic agents acting under resource-constraints has been investigated by  researchers; many of these approaches extend the alternating-time temporal logic (\ATL)~\cite{Alur//:02a} with actions that, in the general case, consume or produce resources. If no bound on the possible amount of resources is given the model checking problems are easily undecidable~\cite{Bulling/Farwer:10a}. Exceptions are possible if restrictions are imposed on the language~\cite{Bulling15ral-IJCAI} or on the semantics~\cite{Alechina//:14c,DellaMonica//:13a}. In many settings,  resources are private to agents, each agent has its own set of resources. In~\cite{DellaMonica//:13a}  resources are shared and a resource money is used to claim resources. The authors present a decidable model checking result which is possible as the amount of resources is bounded. In this paper we are interested in the model checking problem where resources are shared and \emph{unbounded}; resources can be consumed and produced without an upper bound on the total number of resources. The setting is rather natural. Resources are shared in e.g.,  the travel budget of a computer science department. All departmental members compete for the travel budget. Parts of the travel money of a successful grant application will be credited to the department's budget; there is no a priori bound on the total budget.

In this paper, we show that the model checking problem for the resource agent logic $\ral$~\cite{Bulling/Farwer:10a} considered here is undecidable in general  when there are more than two of such unbounded shared resources. This result follows as a corollary from~\cite{Bulling/Farwer:10a,Bulling15ral-IJCAI} where model checking resource bounded systems with private, unbounded resources has been proved undecidable. 
Secondly, we show that model checking  \ral is decidable in case of  a single shared, unbounded resource. Although this seems intuitive, as a single unbounded resource can intuitively  be encoded by a single stack/counter, its proof is (technically)  non-trivial and is based on a reduction to  alternating B\"uchi pushdown systems~\cite{song2014efficient,bouajjani1997reachability}. We first introduce  compact alternating B\"uchi pushdown systems (CABPDSs)  to encode the resource bounded models of our logic such  that the runs of the automaton can be related to execution trees of a given set of agents in the model. We show that model checking \CTL over these systems is decidable using results of~\cite{song2014efficient}. Finally, we reduce model checking \ral to model checking \CTL over CABPDSs. These results extend work on model checking \CTL over pushdown systems where atomic propositions can be given by regular languages~\cite{song2014efficient}. The latter results, in turn, are based on~\cite{bouajjani1997reachability} where reachability of alternating pushdown systems and model checking problems over pushdown systems with standard labelling functions are investigated. Model checking \CTL  over pushdown systems and its computational complexity have also been considered in~\cite{bozzelli2007complexity}. Our model checking problem is also related to reachability in B\"uchi games~\cite{cachat2002symbolic}. Many complexity results about ABPDSs and their variants are known and established in the above mentioned pieces of work. In our future research we  plan to determine the exact computational complexity of the model checking problem for resource agent logic (\ral) over  $1$-unbounded resource bounded models.

The paper is organised as follows. In Section~\ref{sec:logic} we \full{discuss different resource types and} introduce our version of resource agent logic with shared resources.  In Section~\ref{sec:ABPDS} we recall  alternating B\"uchi pushdown systems (ABPDSs) and variants thereof. We propose compact ABPDSs for encoding our models. We show that model checking \CTL over them is decidable. In Section \ref{sec:decidable} we give our main decidability result for a single unbounded resource\short{ and also conclude undecidability for the general setting with more than 1-unbounded shared resource}. Finally, \full{in Section~\ref{sec:undecidable} we consider the general case and show that model checking is undecidability if at least two unbounded resource types are available, and }\short{we }conclude in Section~\ref{sec:concl}. \short{Due to lack of space, we have to skip details and proofs. An extended version of this paper  can be found  in~\cite{Bulling15ABPDS-xarchive}.}
\vspace{-0.4cm}

\section{Resource Agent Logic} \label{sec:logic} 
In this section we define the logic \emph{resource agent logic} \ral and resource-bounded models. The framework is essentially based on~\cite{Bulling15ral-IJCAI}. \full{We begin with a discussion on different resource types which can be classified among different dimensions:
\begin{description}
\item \emph{Private} resources are assigned to individual agents.
\item \emph{Shared} resources can be accessed by all agents; they are global.
\item \emph{Consumable/producible} resources can be consumed and produced. They often disappear after usage, like gasoline and energy, and may thus also be  labelled \emph{fluent}.
\item \emph{Re-usable} resources do not in general disappear after usage. They may also be produced.
\item \emph{Bounded/unbounded} resource types characterise whether arbitrarily many resources can be produced or if there is a bound on the maximal amount.
\end{description}
We note that the property of boundedness has a different flavour in comparison to the other properties. It is better understood as a property of the agents or of the specific modelling rather than of the resource itself. For example, the agent can only carry up to two heavy boxes, or there are legislations which prohibit to have more than three cars in a household. In this paper we are interested in shared, unbounded, consumable  resource types. That is, there is a common pool of resources for which all agents compete. Agents' actions may consume resources or produce them, always affecting the common pool of resources. Moreover, there can be arbitrarily many resources of a resource type.

Clearly, in real settings there will usually be a combination of different resource types. Adding bounded resources will in general not affect the decidability of model checking. Such resources can be encoded in the states, blowing up the model. One has to be more careful with unbounded resources. Having at least two  unbounded resource types is often a first indication for an  undecidable model checking problem if no other restrictions are imposed on the  setting.} 
%
\short{But, in this paper we are intersted in \emph{shared resources}. There is a common pool of resources and agents compete for them. There are more dimensions along which resources can be classified~\cite{Bulling15ABPDS-xarchive}, one of these dimensions is \emph{boundedness}. A resource is called \emph{unbounded} if there is no a priori bound on the number of available resources, in principle they can be produced without limit. Settings with only bounded resources are often  decidable~\cite{Bulling/Farwer:10a}. In the following we consider unbounded resources and assume that  $\Resources$ is a finite non-empty set of such unbounded, shared resource types. } 
\full{Following the discussion above, we define a}\short{A }  \emph{(shared) endowment (function)} $\eta : \Resources \rightarrow \nat_0$
\full{to specify}\short{specifies} the available shared resources of the resource types  $\Resources$ in the system; i.e., $\eta(r)$ is the number of shared resources of type $r$.  With $\Enments$ we denote the set of all possible endowments.  A special minimal endowment function is denoted by $\bar0$. It expresses that there are no resources at all. 

\full{\begin{definition}[Shared resource structure, unbounded]
A \emph{(shared) resource structure} is a tuple $\mathfrak{R}=(\Resources,\beta)$ where $\Resources$ is a finite set of shared  consumable resources. Function $\beta: \Resources\rightarrow \mathbb{N}\cup\{\infty\}$ is called \emph{resource bound}. It specifies the maximal number of resources of a specific type in the model. We say that $\mathfrak{R}$ is \emph{$k$-unbounded} iff the number of unbounded resource types is at most $k$.
\end{definition}}

\textbf{Syntax.} Resource agent logic (\ral) is defined over a set of agents $\Agents$ and a set of propositional symbols $\Props$.  
\ral-formulae\footnote{Note that we slightly change the notation in comparison with~\cite{Bulling/Farwer:10a} where $\coopdown{A}$ has the meaning of $\coopdown[\Agt]{A}$. Moreover, we only use operators that refer to the currently available resources in the system.}
are essentially generated according to the grammar of $\mathsf{ATL}$~\cite{Alur//:02a} 
as follows: $\phi::= \prop{p} \mid \neg \phi \mid \phi \wedge \phi \mid 
         \coopdown{A}\Next \varphi \mid  
         \coopdown{A}\varphi \NUntil \psi  \mid        
         \coopdown{A}\NBox \varphi$ 
%
         where $\prop{p}\in\Props$ is a proposition and $A\subseteq \Agt$ is a set of agents. 
        
        A formula  $\coopdown{A}\varphi$ is called \emph{flat} if $\varphi$ contains no  cooperation modalities. 
         The operators $\Next$, $\NUntil$, and $\NBox$ denote the standard temporal operators expressing that some property holds in the \emph{next} point in time,  \emph{until} some other property holds, and \emph{now and always} in the future, respectively.  The \emph{eventually} operator is defined as macro: $\Sometm\varphi=\top\NUntil\varphi$ (\emph{now or sometime in the future}).  
         The  cooperation modality  $\coopdown{A}$  assumes that \emph{all} agents  in \Agents act under resource constraints. The reading of $\coopdown{A}{}\varphi$ is that
           \emph{agents $A$ have a strategy compatible with the currently available resources} to enforce $\varphi$. 
           This  means that the strategy can be executed given the agents' resources. Thus, it is necessary to keep track of resource production and consumption during the execution of a strategy. 


\textbf{Semantics.} We define the models of \ral as in~\cite{Bulling15ral-IJCAI}. We also introduce a special class of these models in which agents have an \emph{idle action} in their repertoire that neither consumes nor produces resources. Note that a model with idle actions is a special case of the general model.

\vspace{-1ex}
\begin{definition}[\rbm, \rbmi, unbounded]
A resource-bounded model (\rbm) is given by $\model =
(\Agt,Q,\Pi,\pi,\Act,d,o,\full{\mathfrak{R}}\short{\Resources},t)$ where \full{$\mathfrak{R}=(\Resources,\beta)$ is a shared resource structure}\short{$\Resources$ is a set of shared, unbounded resources}, $\Agt=\{1,\ldots,k\}$ is a set of agents; $\pi : \Props \to \powerset{\States}$ is a valuation of
propositions; $\Act$ is a finite set of actions; and the function $d : \Agt \times Q
\to \powerset{\Act}\backslash\{\emptyset\}$ indicates the actions available to agent $a \in \Agt$ at state $q
\in Q$. We write $d_a(q)$ instead of $d(a, q)$, and use $d(q)$ to denote the set
$d_1(q) \times \ldots \times d_k(q)$ of action profiles in state $q$. Similarly, $d_A(q)$ denotes the action tuples available to $A$ at $q$. $o$ is a
 transition function which maps each state $q \in Q$ and action profile
$\vec{\alpha} = (\alpha_1, \ldots, \alpha_k) \in d(q)$ (specifying a
move for each agent) to another state $q' = o(q,\vec{\alpha})$. Finally, the
function $t : \Act \times \Resources \to \integer$ models the resources consumed and
produced by actions. We define $\production(\alpha,r) := \max\{0, t(\alpha,
r)\}$ (resp.\ $\consumption(\alpha, r) := \min\{0, t(\alpha, r)\}$) as the
amount of resource $r$ produced (resp.\ consumed) by action $\alpha$.  For $\vec{\alpha} = ( \alpha_1,\ldots,\alpha_k)$, we use
$\vec{\alpha}_A$ to denote the sub-tuple consisting of the actions of agents $A
\subseteq \Agt$.\short{ We call an \rbm \emph{$k$-unbounded} if $|\Resources|=k$ for a natural number $k$.} 

An \emph{\rbm with idle actions}, $\rbmi$ for short, is an \rbm \model such that for all agents $a$, all states $q$, there is an action $\alpha\in d_a(q)$ such that for all resource types $r$ in \model we have that  $t(\alpha,r)=0$. We refer to this action (or to one of them if there is more than one) as the \emph{idle action} of $a$ and denote it by $idle$.
\end{definition}

\vspace{-2ex}
A \emph{path}  
$\lambda \in Q^\omega$ is an infinite sequence
of states such that there is a transition between two adjacent states. A
\emph{resource-extended} path $\lambda \in 
(Q\times\Enments)^\omega$ is an infinite sequence over
$Q\times\Enments$ such that the restriction to states (the first component),
denoted by $\lambda|_Q$, is a path in the underlying model.  The projection of
$\lambda$ to the second component of each element in the sequence is denoted by
$\lambda|_\Enments$. 
We define $\lambda[i]$ to be the $i+1$-th element of $\lambda$, and
$\lambda[i,\infty]$ to be the suffix $\lambda[i]\lambda[i+1]\ldots$. 
A \emph{strategy}\footnote{{We note that differently from~\cite{Bulling/Farwer:10a,Alechina//:14c,Bulling15ral-IJCAI}, our notion of strategy takes the history of states as well as the history of endowments into account. In the setting considered here such strategies are more powerful than strategies only taking the state-component into account.}} for a coalition $A \subseteq \Agt$ is a function 
$s_A : (Q\times\Enments)^+ \to \Act^A$ such that $s_A((q_0,\enment_0)\ldots(q_n,\enment_n)) \in d_A(q_n)$ for $(q_0,\enment_0)\ldots(q_n,\enment_n)\in (Q\times\Enments)^+$. Such a strategy gives rise to a set of (resource-extended) paths that can emerge if agents  follow their  strategies. A \emph{$(q,\eta,s_A)$-path} is a 
resource-extended path $\lambda$ such that for all $i=0,1,\ldots$
with $\lambda[i]:=(q_i,\eta_i)$ there is an action profile
$\vec{\alpha} \in d(\lambda|_\States[i])$ such that:
\full{\begin{enumerate}}
\full{\item }\short{(i) }$q_0=q$ and  $\eta_0(r)\short{=\enment(r)}\full{ = \min\{\beta(r),\enment(r)\}}$ for all $r\in\Resources$ (describes  initial configuration); 
\full{\item }\short{(ii) } $s_A(\lambda[0,i])=  \vec{\alpha}_A$ 
 ($A$ follow their strategy); 
\full{\item }\short{(iii) } $\lambda|_Q[i+1]=o(\lambda|_Q[i],\vec{\alpha})$ 
 (transition according to $\vec{\alpha}$); 
\full{\item }\short{(iv) } for all $\vec{\alpha}'\in \Act_{\Agt\backslash A}$ and for all $r\in\Resources$: $\eta_{i}(r) \geq \sum_{a \in \Agt\backslash A}\consumption(\vec{\alpha}'_a,r)+\sum_{a \in A}\consumption(\vec{\alpha}_a,r)$ 
 (enough resources to perform the actions are available); and
 \full{\item}\short{(v)}   $\eta_{i+1}(r)=\eta_i(r)+\sum_{a\in\Agents}\production(\vec{\alpha}_a) -\sum_{a\in\Agents}\consumption(\vec{\alpha}_a)$ for all $r\in\Resources$.
\full{\end{enumerate}} Condition (iv) models that the opponents have priority when claiming resources.
The \emph{$(q,\eta,s_A)$-outcome} of a strategy $s_{A}$ in $q$,
$\rhooutcome{q}{\eta}{s_{A}}$, is defined as the set of all  $(q,\eta,s_{A})$-paths
starting in $q$. We also refer to this set as an \emph{execution tree} of $A$.  Truth is defined over an \rbm $\model$, a state $q\in\States$, and an endowment
$\eta$. The \emph{semantics} is given by the satisfaction relation
$\modelsR$\full{ defined below}. \short{Here, we only present clauses for two types of formulae: $\model,\state,\enment \modelsR\prop{p}$
  iff $\prop{p}\in\Props$ and $q\in\pi(\prop{p})$; and $\model,q,\enment \modelsR \coopdown{A}\psi\NUntil\varphi$
    iff  there exists a strategy $s_A$ for $A$ such that for all $\onepath\in out(\state,\enment,s_A)$, there is an $i$ with $i\geq 0$ 
    and 
    $\model,\onepath|_Q[i],\onepath|_\Enments[i] \modelsR \varphi$ such that 
    for all $j$ with $0 \leq j < i$ it holds that 
    $\model,\onepath|_Q[j],\onepath|_\Enments[j] \modelsR \psi$. The other clauses are given analogously, cf.~\cite{Bulling15ABPDS-xarchive}. }
\full{\begin{description}
\item[$\model,\state,\enment \modelsR\prop{p}$]
  iff $\prop{p}\in\Props$ and $q\in\pi(\prop{p})$.
 \item[$\model,\state,\enment \modelsR\varphi_1\wedge\varphi_2$]
   iff $\model,\state,\enment \modelsR\varphi_1$ and $\model,\state,\enment \modelsR\varphi_2$
 \item[$\model,\state,\enment \modelsR\neg \varphi$]
   iff it is not the case that  $\model,\state,\enment \modelsR\varphi$
\item[$\model,\state,\enment \modelsR{\coopdown{A}}\Next\varphi$]
  iff there is a strategy $s_A$ for $A$ such that
  for all $\onepath\in out(\state,\enment,s_A)$, 
  $\model,\onepath_Q[1],\enment \modelsR \varphi$

\item[$\model,q,\enment \modelsR \coopdown{A}\psi\NUntil\varphi$]
  iff  there exists a strategy $s_A$ for $A$ such that for all $\onepath\in out(\state,\enment,s_A)$, there is an $i$ with $i\geq 0$ 
  and 
  $\model,\onepath|_Q[i],\onepath|_\Enments[i] \modelsR \varphi$ such that 
  for all $j$ with $0 \leq j < i$ it holds that 
  $\model,\onepath|_Q[j],\onepath|_\Enments[j] \modelsR \psi$
  
 \item[$\model,q,\enment \modelsR\coopdown{A}\NBox\varphi$]
   iff there exists a strategy $s_A$ for $A$ such that for all $\onepath\in out(\state,\enment,s_A)$ and 
   all $i\geq 0$, 
   $\model,\onepath|_Q[i],\onepath|_\Enments[i] \modelsR \varphi$ 

\end{description}}
The  \emph{model checking problem} is to determine whether  $\model,q,\enment\models\varphi$ holds. 

\textbf{Example.} We illustrate the framework by extending  the introductory example on the departmental 
travelling budget. Consider a  department
which consists of a dean $d$, two professors $p_1, p_2$ and three lecturers 
$l_1, l_2,$ and $l_3$. The department's  travel budget 
is allocated  annually and can be spent 
to attend  conferences. 
There
are three categories to request money: premium, 
advanced, and economic. All options are available to the dean,
the last two to professors, and only the last one to the lecturers.
For instance, if the cost of attending PRIMA\full{\footnote{PRIMA is the acronym for the conference Principles and Pratice of Multi-Agent Systems. A short version of this paper was accepted for PRIMA 2015~\cite{Bulling15ABPDS}.}} is, depending on the category, \$2000, \$1000, and \$500,  
respectively, then with an available budget of \$4000 not all lecturers can be sure to be able to attend PRIMA. Because,  the dean and the professors could all decide to attend PRIMA and to request the advanced category. In that case, only \$1000 would remain, not enough  for all lecturers to attend; formally specified,  
$\coopdown{d,p_1,p_2}\Sometm(\prop{d}\wedge \prop{p_1}\wedge\prop{p_2}\wedge\neg \coopdown{\{l_1,l_2,l_3\}}\Sometm(\prop{l_1}\wedge \prop{l_2}\wedge\prop{l_3}))$ is true where a proposition $\prop{x}$ expresses that ``person'' $x$ is attending PRIMA. Equivalently, $\neg \coopdown{\{l_1,l_2,l_3\}}\Sometm(\prop{l_1}\wedge \prop{l_2}\wedge\prop{l_3}))$ is true; this highlights that the opponents have priority in claiming resources. However, by collaborating with the professors, they have a strategy which allows all lecturers to attend, independent of the actions of the dean: 
i.e., $\coopdown{\{p_1,p_2,l_1,l_2,l_3\}}\Sometm (\prop{l_1}\wedge \prop{l_2}\wedge\prop{l_3})$.

\section{Model Checking CTL over B\"uchi Pushdown Systems}\label{sec:ABPDS}
\newcommand{\PDS}{\ensuremath{\mathcal{P}}\xspace}
\newcommand{\APDS}{\ensuremath{\mathcal{P}}\xspace}
\newcommand{\AAut}{\ensuremath{\mathcal{A}}\xspace}
\newcommand{\Lc}{\ensuremath{\widehat{L}}\xspace}
\newcommand{\ABPDS}{\ensuremath{\mathcal{B}}\xspace}
\newcommand{\Pre}{\ensuremath{\mathsf{Pre}}\xspace}
\newcommand{\Cnf}{\ensuremath{\mathsf{Cnf}}\xspace}
\newcommand{\Post}{\ensuremath{\mathsf{Post}}\xspace}
\newcommand{\cl}{\ensuremath{\mathsf{cl}}\xspace}
\newcommand{\lab}{\ensuremath{\mathsf{lab}}\xspace}
\newcommand{\Release}{\ensuremath{\mathbf{R}}\xspace}
\newcommand{\CABPDS}{\ensuremath{\mathcal{C}}\xspace}

We first review existing results on alternating B\"uchi pushdown systems (ABPDSs).  Then, we use these results to give an automata-theoretic approach to model check \LogCTL-formulae over \emph{compact} ABPDSs.  The latter will be used to encode \rbm{}s in Section~\ref{sec:decidable}. 
\full{ An \emph{alphabet} $\Gamma$ is a non-empty, finite set of symbols. $\Gamma^*$ denotes the set consisting of all finite words over $\Gamma$ including the empty word $\epsilon$. Typical symbols from $\Gamma$ are denoted by $a,b,\ldots$ and words by $w,v,u,\ldots$. We read words from left to right. As before, we assume that  $\Props$ denotes a finite, non-empty set of propositions.}
\subsection{Alternating B\"uchi Pushdown Systems}
\short{An \emph{alphabet} $\Gamma$ is a non-empty, finite set of symbols. $\Gamma^*$ denotes the set consisting of all finite words over $\Gamma$ including the empty word $\epsilon$. Typical symbols from $\Gamma$ are denoted by $a,b,\ldots$ and words by $w,v,u,\ldots$. We read words from left to right. As before, we assume that  $\Props$ denotes a finite, non-empty set of propositions. } We  use words to represent the stack content. We say that word $w=a_1\ldots a_n$ is on the stack if $a_1$ is the lowest symbol, followed by $a_2$ and so forth. The symbol on top is $a_n$. An \emph{alternating pushdown system} (APDS) is a tuple $\APDS=(P,\Gamma,\Delta)$ where $P$ is a non-empty, finite set of control states, $\Gamma$ a non-empty, finite (stack) alphabet, and $\Delta\subseteq(P\times \Gamma)\times \powerset{P\times\Gamma^*}$ a transition relation~\cite{bouajjani1997reachability,suwimonteerabuth2006efficient-tr}. We call \APDS a \emph{pushdown system} (PDS) if $(s,a)\Delta X$ implies $|X|=1$ where $X\in\powerset{P\times\Gamma^*}$. An \emph{alternating B\"uchi pushdown system} (ABPDS) $\ABPDS=(P,\Gamma,\Delta,F)$ is defined as a APDS but a set of accepting states $F\subseteq P$ is added. 
 In the following we  focus on ABPDSs, but most of the definitions do also apply to APDSs and PDSs with obvious changes. A transition $(p,a)\Delta\{(p_1,w_1),\ldots,(p_n,w_n)\}$ represents that if the system is in state $p$ and the top-stack symbol is $a$ then the ABPDS \ABPDS is copied $n$-times where the $i$th copy changes its local state to $p_i$, pops $a$ from the stack and pushes $w_i$ on the stack, $1\leq i\leq n$. For a  transition rule $(p,a)\Delta\{(p_1,w_1),\ldots,(p_n,w_n)\}$ and a stack content $w\in\Gamma^*$ we say that $(p,wa)$ is an \emph{immediate predecessor} of $\{(p_1,ww_1),\ldots,(p_n,ww_n)\}$. We write $(p,wa)\Rightarrow_\ABPDS\{(p_1,ww_1),\ldots,(p_n,ww_n)\}$.  \full{We also say that $\{(p_1,ww_1),\ldots,(p_n,ww_n)\}$ is an \emph{immediate successor} of $(p,wa)$. We often write $(p,a)\Delta(p',w)$ for $(p,a)\Delta\{(p',w)\}$. Finally, we  would like to note that a \emph{stack bottom symbol} $\#$ can be defined the only purpose of which is to denote that the stack is empty. Apart from this the symbol is never touched. The introduction of $\#$ simply requires adding $\#$ to $\Gamma$ and to add a rules which pushes $\#$ to the stack, before any other rule is applied. In the following we assume that $\#$ is the stack bottom symbol whenever it appears in the text.} 
\full{\\ }
A \emph{configuration} of $\ABPDS$ is a tuple from $\Cnf_\ABPDS=P\times\Gamma^*$. A \emph{$c$-run $\rho$} of \ABPDS, where  $c$ is a configuration of $\ABPDS$, is a tree in which  each node is labelled by a configuration such that the root of the tree is labelled by $c$.  If a node labelled by $(p,w)$ has $n$ (direct) child nodes labelled by $(p_1,w_1),\ldots,(p_n,w_n)$, respectively,  then it is required that $(p,w)\Rightarrow_\ABPDS\{(p_1,w_1),\ldots,(p_n,w_n)\}$. We use $\Runs_\ABPDS(c)$ to denote the set of all $c$-runs\full{\footnote{{Sometimes, we assume that elements in $X$ in  a transition $(p,a)\Delta X$ are ordered and correspondently the branches in a run.}}} and $\Runs_\ABPDS=\bigcup_{c\in\Cnf_\ABPDS}\Runs_\ABPDS(c)$. We note that a run in a PDS \PDS is simply a linear sequence of configurations. A \emph{$\rho$-path},  $\rho\in\Runs_\ABPDS(c)$,  is a maximal length branch $\kappa=c_0c_1\ldots$ of $\rho$ starting at the root node $c$. We shall identify $\rho$ with its set of paths and write $\kappa\in\rho$ to indicate that $\kappa$ is a $\rho$-path. Again, in the case of a PDS \PDS a run and a path in it are essentially the same. 
We say that $\kappa\in\rho$ is \emph{accepting} if a state of $F$ occurs infinitely often in configurations on $\kappa$. A run is accepting if each path $\kappa\in\rho$  is accepting; and a configuration $c$ is accepting if there is an accepting run $\rho\in\Runs_\ABPDS(c)$. \full{We note that an accepting run of an ABPDS has only infinite branches. }The \emph{language} accepted by $\ABPDS$, $L(\ABPDS)$, is the set of all accepting configurations. \full{Finally, we define $\Rightarrow^*_\ABPDS\subseteq(P\times\Gamma^*)\times\powerset{P\times\Gamma^*}$ as, roughly speaking, the reflexive transitive closure of $\Rightarrow_\ABPDS$; that is, $c\Rightarrow^*_\ABPDS\{c\}$ for all $c$; if $c\Rightarrow_\ABPDS C$ then $c\Rightarrow^*_\ABPDS C$; and if $c\Rightarrow^*_\ABPDS \{c_1,\ldots,c_n\}$ and $c_i\Rightarrow^*_\ABPDS C_i$ for every $1\leq i\leq n$, then $c\Rightarrow^*_\ABPDS\bigcup_i C_i$. 
}

A nice property of an ABPDS is that its set of accepting configurations is regular, in the sense that it is accepted by an appropriate automaton which is defined next. An \emph{alternating automaton}~\cite{bouajjani1997reachability} is a tuple $\AAut=(S,\Sigma,\delta, I, S_f)$ where  $S$ is a finite, non-empty set of states,  $\delta\subseteq S\times \Sigma\times\powerset{S}$ is a transition relation, $\Sigma$ an input alphabet, $I\subseteq S$ a set of initial states, and $S_f\subseteq S$ a set of final states. \full{Similar to $\Rightarrow^*_\ABPDS$ we define the reflexive, transitive transition relation $\rightarrow^*_\AAut\subseteq (S\times\Sigma^*)\times \powerset{S}$ as follows (where we write $s\stackrel{w}{\rightarrow}^*_\AAut S'$ for $(s,w,S')\in\rightarrow^*_\AAut$): $s\stackrel{\epsilon}{\rightarrow}^*_\AAut\{s\}$, if $(s,a,S')\in\delta$ then $s\stackrel{a}{\rightarrow}^*_\AAut S'$, and if $s\stackrel{w}{\rightarrow}^*_\AAut\{s_1,\ldots,s_n\}$ and $s_i\stackrel{a}{\rightarrow}^*_\AAut S_i$ for $1\leq i\leq n$ then $s\stackrel{wa}{\rightarrow}^*_\AAut\bigcup_iS_i$.}  The automaton \emph{accepts} $(s,w)\in S\times\Sigma^*$ iff \full{$s\stackrel{w}{\rightarrow}^*_\AAut S'$  with $S'\subseteq S_f$ and} $s\in I$\short{ and each state reached after the automaton has read $w$ is a final state}. The  language accepted by $\AAut$ is denoted by $L(\AAut)$. A language is called \emph{regular} if it is accepted by an alternating automaton. \full{Finally, for a given  ABPDS $\ABPDS=(P,\Gamma,\Delta,F)$ we define an \emph{alternating \ABPDS-automaton}  as an alternating automaton $(S,\Sigma,\delta, I, S_f)$ such that $I\subseteq P\subseteq S$ and $\Gamma=\Sigma$.} We recall the following result from~\cite{song2014efficient}:

\vspace{-1ex}
\begin{theorem}[\cite{song2014efficient}]\label{theo:song1}
For any  ABPDS \ABPDS there is an effectively computable alternating $\ABPDS$-automaton $\AAut$ such that $L(\AAut)=L(\ABPDS)$.
\end{theorem}
\full{The authors of~\cite{song2014efficient} do also determine the size of the automaton. As we are not concerned with the computational complexity in this paper, we omit these results.}

\subsection{Model Checking CTL over \full{PDSs}\short{ABPDSs}}\label{sec:CTL}
\full{\textbf{The Logic CTL.} Computation Tree Temporal Logic (\CTL)~\cite{Clarke81ctl} can be seen as the one agent, non-resource-constrained variant of $\ral$.   Formulae of the logic are defined by the grammar: $\phi::= \prop{p} \mid \neg\phi\mid \phi \wedge \phi \mid 
\Epath\Next \varphi\mid \Epath\Always \varphi \mid \Epath( \varphi \NUntil \psi)$ 
where $\prop{p}\in\Props$. 
 $\Epath$  denotes the existential  path quantifier. $\Epath\varphi$ expresses that there is a run on which $\varphi$ holds. The Boolean connectives are given by their usual abbreviations. In addition to that, we define the  macros  $\Sometm\varphi\equiv \true\NUntil\varphi$, $\Apath\Next\varphi\equiv\neg\Epath\Next\neg\varphi$, $\Apath\Always\varphi\equiv\neg\Epath\Sometm\neg\varphi$, and $\Apath\varphi\NUntil\psi\equiv \neg\Epath((\neg\psi)\NUntil(\neg\varphi\wedge\neg\psi))\wedge\neg\Always\neg\psi$. Thus, $\Apath\varphi$ is read as $\varphi$ holds on all runs. The other temporal operators  have  the same meaning as for \ral. Moreover, for our constructions it is  assumed that \CTL-formulae are in  \emph{negation normal form}, that is negation only occurs at the propositional level. This makes it necessary to allow  the connective $\vee$ (\emph{or}) and also the \emph{release operator} $\Release$ as first-class citizens in the object language. Therefore, we use the following macros: $\Apath\varphi_1\Release\varphi_2=\neg\Epath(\neg\varphi_1)\NUntil(\neg\varphi_2)$ and $\Epath\varphi_1\Release\varphi_2=\neg\Apath(\neg\varphi_1)\NUntil(\neg\varphi_2)$. A subformula of a formula $\varphi$ is a formula that occurs in $\varphi$, including $\varphi$ itself. The \emph{closure of $\varphi$}, $\cl(\varphi)$, is the set of all subformulae of $\varphi$. We  define the set  $\Props^+(\varphi)=\{\prop{p}\in\Props\mid \prop{p}\in\cl(\varphi)\}$ and $\Props^-(\varphi)=\{\prop{p}\in\Props\mid \neg\prop{p}\in\cl(\varphi)\}$ containing all propositional variables that occur positively and negatively in $\varphi$, respectively. Later, we also need a special closure $\cl_\Release(\varphi)$ which consists of all formulae of the form $\Apath(\varphi_1\Release\varphi_2)$ or $\Epath(\varphi_1\Release\varphi_2)$ of $\cl(\varphi)$.}

\full{\textbf{Model Checking over Pushdown Systems.}}
\short{We first consider model checking \emph{Computation Tree Logic} (\LogCTL) over PDSs. We assume that the reader is familiar with $\LogCTL$ and refer to~\cite{Clarke81ctl,Bulling15ABPDS-xarchive} for details. Essentially, the cooperation modalities of \ral are replaced by the existential and universal path quantifiers $\Epath$ and $\Apath$, respectively. The formula $\Epath\varphi$ expresses that there is a path along which $\varphi$ holds; analogously, $\Apath\varphi$ expresses that $\varphi$ holds along all paths. }
The problem of \CTL model checking over PDSs has been considered in, e.g.,  \cite{song2014efficient,bouajjani1997reachability,bozzelli2007complexity}. 
We now recall from~\cite{song2014efficient} how the problem is defined. First, the PDS is extended with a labelling function $\lab$ to give truth to propositional atoms. In~\cite{song2014efficient} two alternatives are considered.  The first alternative  assigns states to propositions, $\lab:\Props\rightarrow\powerset{P}$. The second alternative   assigns configurations to propositions, $\lab:\Props\rightarrow\powerset{P\times\Gamma^*}$. In the following we only consider the second, more general alternative as this is the one we shall need for  model checking \ral. For this type of labelling function we need a finite representation. We call $\lab$  \emph{regular} if there is an alternating automaton $\AAut_{\prop{p}}$ with $L(\AAut_{\prop{p}})=\lab(\prop{p})$ for each $\prop{p}\in\Props$. We are ready to give the semantics of \CTL-formulae over a  PDS $\PDS=(P,\Gamma,\Delta)$, $c\in\Cnf_\PDS$, and a regular labelling function $\lab:\Props\rightarrow\powerset{P\times\Gamma^*}$. \short{The semantic clauses are defined in the usual way, due to space contraints we present one clause only and refer to~\cite{Bulling15ABPDS-xarchive} for further details: 
$\PDS,c,\lab\models \Epath\Always\varphi$ iff there is a $c$-run $\rho=c_0c_1,\ldots\in\Runs_\PDS(c)$ such that $\PDS,c_i,\lab\models\varphi$ for all $i\geq 0$.
} 
\full{The semantics is defined by $\models$ as follows: 
\begin{description}
\item $\PDS,c,\lab\models \prop{p}$ iff $c\in\lab({\prop{p}})$:
\item $\PDS,c,\lab\models \neg\varphi$ iff it is not the case that $\PDS,c,\lab\models \varphi$:
\item $\PDS,c,\lab\models \varphi_1\wedge\varphi_2 $ iff it is the case that $\PDS,c,\lab\models \varphi_1$ and $\PDS,c,\lab\models \varphi_2$:
\item $\PDS,c,\lab\models \Epath\Next\varphi$ iff there is a $c$-run $\rho=c_0c_1,\ldots\in\Runs_\PDS(c)$ such that  $\PDS,c_1,\lab\models\varphi$:
\item $\PDS,c,\lab\models \Epath\Always\varphi$ iff there is a $c$-run $\rho=c_0c_1,\ldots\in\Runs_\PDS(c)$ such that $\PDS,c_i,\lab\models\varphi$ for all $i\geq 0$:
\item $\PDS,c,\lab\models \Epath\varphi_1\NUntil\varphi_2$ iff   there is a $c$-run $\rho=c_0c_1,\ldots\in\Runs_\PDS(c)$ such that  there is an $i\in\mathbb{N}_0$ with $\PDS,c_i,\lab\models\varphi_2$ and for all $0\leq j<i$ we have that $\PDS,c_j,\lab\models\varphi_1$.
\end{description}}

The authors of~\cite{song2014efficient} give a model checking algorithm which uses ABPDSs.  They construct from $\PDS$,  $\lab$ and $\varphi$, an ABPDS $\ABPDS_{\PDS,\varphi}$ such that $\PDS,(p,w),\lab\models\varphi$ iff $((p,\varphi),w)\in L(\ABPDS_{\PDS,\varphi})$. The ABPDS  is essentially the product of the PDS \PDS with the closure\short{ $\cl(\varphi)$} of $\varphi$\short{\footnote{The closure $\cl(\varphi)$ is the set of all subformulae of $\varphi$.}}, in particular states of $\ABPDS_{\PDS,\varphi}$ are tuples $(p,\psi)\in P\times\cl(\varphi)$.  The existential and universal path quantifiers of the formula cause the alternation of the ABPDS. \full{We will give more details in Section~\ref{sec:CTLABPDS} where we consider model checking \CTL-formulae over ABPDSs. We finish this section by recalling the following theorem which follows from Theorem~\ref{theo:song1}.

\begin{theorem}[\cite{song2014efficient}]\label{theorem:song2}
For a given PDS \PDS, a regular labelling function $\lab$, and a \CTL-formula $\varphi$ there is an effectively computable alternating automaton $\AAut_{\PDS,\varphi}$  such that for all configurations $c=(p,w)\in\Cnf_\PDS$  the following holds: $\PDS,c,\lab\models\varphi$ iff $((p,\varphi),w)\in L(\AAut_{\PDS,\varphi})$.
\end{theorem}}

\full{\subsection{Model Checking CTL over \short{Compact} ABPDS}\label{sec:CTLABPDS}}
For our later results, we need to be able to define the truth of \CTL-formulae over ABPDSs rather than PDSs. \full{We extend the result of Theorem~\ref{theorem:song2} accordingly.} Let an ABPDS \ABPDS be given. We first discuss what it means that $\ABPDS,c,\lab\models \Epath\varphi$. As before, we interpret it as:  there is a \emph{run} $\rho\in\Runs_\ABPDS(c)$ on which $\varphi$ holds. However, given that $\rho$ is a tree 
in the case of ABPDSs (or a set of paths) we need to explain how to evaluate $\varphi$ on trees. We require that $\varphi$ must be true on \emph{each} path $\kappa\in\rho$ on the run. This can nicely be illustrated if \ABPDS is considered as a two player game where player one decides which transition to take, and player two selects one of the child states. Thus, $\Epath\varphi$ expresses that player one has a strategy in the sense that it can enforce a run $\rho$ such that player two cannot make $\varphi$ false on any path $\kappa\in\rho$. \short{For a given configuration $c\in\Cnf_\ABPDS$, and a regular labelling $\lab$, the semantic rules have the following form~\cite{Bulling15ABPDS-xarchive}: 
$\ABPDS,c,\lab\models \Epath\Always\varphi$ iff there is a $c$-run $\rho\in\Runs_\ABPDS(c)$  such that for all paths $c_0c_1\ldots\in\rho$  it holds that $\ABPDS,c_i,\lab\models\varphi$ for all $i\geq 0$.
}\full{The precise semantics is given next where  $c\in\Cnf_\ABPDS$, and $\lab$ is a regular labelling function as before:

\begin{description}
\item $\ABPDS,c,\lab\models \prop{p}$ iff $c\in\lab(\prop{p})$;
\item $\ABPDS,c,\lab\models \neg\varphi$ iff it is not the case that $\ABPDS,c,\lab\models \varphi$;
\item $\ABPDS,c,\lab\models \varphi_1\wedge\varphi_2 $ iff  $\ABPDS,c,\lab\models \varphi_1$ and $\ABPDS,c,\lab\models \varphi_2$;
\item $\ABPDS,c,\lab\models \Epath\Next\varphi$ iff there is a $c$-run $\rho\in\Runs_\ABPDS(c)$ such that for all paths $c_0c_1\ldots\in \rho$ it holds that $\ABPDS,c_1,\lab\models\varphi$;
\item $\ABPDS,c,\lab\models \Epath\Always\varphi$ iff there is a $c$-run $\rho\in\Runs_\ABPDS(c)$  such that for all paths $c_0c_1\ldots\in\rho$  it holds that $\ABPDS,c_i,\lab\models\varphi$ for all $i\geq 0$;
\item $\ABPDS,c,\lab\models \Epath\varphi_1\NUntil\varphi_2$ iff   there is a $c$-run $\rho\in\Runs_\ABPDS(c)$ such that for all paths $c_0c_1\ldots\in\rho$   there is an $i\in\mathbb{N}_0$ with $\ABPDS,c_i,\lab\models\varphi_2$ and for all $0\leq j<i$ we have that $\ABPDS,c_j,\lab\models\varphi_1$.
\end{description}}

\full{\noindent{}We are ready to give a variant of Theorem~\ref{theorem:song2} over ABPDSs.\full{ It follow from  Lemma~\ref{lemma:mcheck-CTL-ABPDS} given below in combination with Theorem~\ref{theo:song1}.}
\begin{theorem}\label{theorem:CTL-ABPDS}
For a given ABPDS \ABPDS, a regular labelling function $\lab$, and a \CTL-formula $\varphi$ there is an effectively computable alternating automaton $\AAut_{\ABPDS,\varphi}$  such that for all configurations $c=(p,w)\in\Cnf_\ABPDS$ the following holds: $\ABPDS,c,\lab\models\varphi$ iff $((p,\varphi),w)\in L(\AAut_{\ABPDS,\varphi})$.

\end{theorem}


The proof of the theorem is closely related to the proof of Theorem~\ref{theorem:song2} given in~\cite{song2014efficient}. Here, however, the branching in the resulting ABPDSs can have two different sources: branching can result from branching in the input ABPDS or from the universal \LogCTL-path quantifier. We sketch the construction of $\ABPDS_{\varphi}$ for a given  ABPDS   $\ABPDS=(P,\Gamma,\Delta,F)$, a regular labelling function $\lab:\Props\rightarrow\powerset{P\times\Gamma^*}$, and a \CTL-formula $\varphi$. 

First, for each $\prop{p}\in\Props$ let $\AAut_{\prop{p}}=(S_\prop{p},\Gamma,\delta_\prop{p},I_\prop{p},F_\prop{p})$ denote the alternating $\ABPDS$-automaton that accepts $L(\AAut_\prop{p})=\lab(\prop{p})$, and $\AAut_{\neg \prop{p}}$ be the alternating automaton with $L(\AAut_{\neg \prop{p}})=P\times\Gamma^*\backslash \lab(\prop{p})$. Due to technical reasons, we make the state spaces disjoint. We add $\prop{p}$ as subindex to every state of $S_\prop{p}$; for example,  a state $p$ becomes $p_\prop{p}$.  Note, that it then holds that $L(\AAut_\prop{p})\subseteq P_\prop{p}\times \Gamma^*$ rather than $L(\AAut_\prop{p})\subseteq P\times \Gamma^*$ where $P_\prop{p}$ is the set of renamed states of $P$. In particular,  the automaton $\ABPDS_{\varphi}$ will include states $(p,\prop{p})$  which are connected to an initial state of the form  $\AAut_{\prop{p}}$. This initial state of $\AAut_\prop{p}$ is denoted by $p_{\prop{p}}$; we proceed similarly for $\AAut_{\neg\prop{p}}$. The ABPDS  $\ABPDS_{\varphi}=(P',\Gamma,\Delta',F')$ is defined as follows (cf. Section~\ref{sec:CTL} for the notation used):
\begin{itemize}
\item $P'=(P\times\cl(\varphi))\cup\bigcup_{\prop{p}\in\Props^+(\varphi)}S_{\prop{p}}\cup\bigcup_{\prop{p}\in\Props^-(\varphi)}S_{\neg\prop{p}}$
\item $F'=(P\times\cl_\Release(\varphi))\cup\bigcup_{\prop{p}\in\Props^+(\varphi)}F_{\prop{p}}\cup\bigcup_{\prop{p}\in\Props^-(\varphi)}F_{\neg\prop{p}}$
\item $\Delta'\subseteq (P'\times\Gamma)\times\powerset{P'\times\Gamma^*}$ is the smallest relation such that:
\begin{enumerate}
\item $((p,\prop{p}),a)\Delta'(p_\prop{p},a)$ for $\prop{p}\in\Props$ 
\item $((p,\neg\prop{p}),a)\Delta'(p_{\neg\prop{p}},a)$  for $\prop{p}\in\Props$ 
\item $((p,\varphi\wedge\psi),a)\Delta'\{(p,\varphi),a),((p,\psi),a)\}$
\item $((p,\varphi\vee\psi),a)\Delta'\{(p,\xi),a)\}$ for $\xi\in\{\varphi,\psi\}$
\item $((p,\Epath\Next\varphi),a)\Delta'\{((p',\varphi),w') \mid (p',w') \in X\}$ for each $(p,a)\Delta X$
\item $((p,\Apath\Next\varphi),a)\Delta'\bigcup_{(p,a)\Delta X}\{((p',\varphi),w') \mid (p',w') \in X\}$
\item $((p,\Epath\varphi\NUntil\psi),a) \Delta'((p,\psi),a)$
\item $((p,\Epath\varphi\NUntil\psi),a)\Delta'\{((p,\varphi),a)\}\cup \{((p',\Epath\varphi\NUntil\psi),w')\mid (p',w') \in X\}$ for each $(p,a)\Delta X$
\item $((p,\Apath\varphi\NUntil\psi),a) \Delta'((p,\psi),a)$
\item $((p,\Apath\varphi\NUntil\psi),a)\Delta'\{((p,\varphi),a)\}\cup \bigcup_{(p,a)\Delta X}\{((p',\Apath\varphi\NUntil\psi),w')\mid (p',w') \in X\}$
\item $((p,\Epath\varphi\Release\psi),a) \Delta'((p,\varphi),a)$
\item $((p,\Epath\varphi\Release\psi),a)\Delta'\{((p,\psi),a)\}\cup\{((p',\Epath\varphi\Release\psi),w')\mid (p',w') \in X\}$ for each  $(p,a)\Delta X$
\item $((p,\Apath\varphi\Release\psi),a) \Delta'((p,\varphi),a)$
\item $((p,\Apath\varphi\Release\psi),a)\Delta'\{((p,\psi),a)\}\cup\bigcup_{(p,a)\Delta X}\{((p',\Apath\varphi\Release\psi),w')\mid (p',w') \in X\}$
\item If $(s,a,S')\in (\bigcup_{\prop{p}\in \Props^+(\varphi)}\delta_\prop{p}\cup \bigcup_{\prop{p}\in \Props^-(\varphi)}\delta_{\neg\prop{p}})$ then   $(s,a)\Delta'\{(s',\epsilon) \mid s' \in S'\}$
\item If $s\in(\bigcup_{\prop{p}\in \Props^+(\varphi)}F_\prop{p}\cup \bigcup_{\prop{p}\in \Props^-(\varphi)}F_{\neg\prop{p}})$ then   $(s,\#)\Delta'(s,\#)$
\end{enumerate}
\end{itemize}
Intuitively, $\ABPDS_{\varphi}$ simulates the semantics of $\CTL$ over ABPDSs and keeps track of the formulae to be satisfied. Let us consider rule 5 and suppose that the current configuration of $\ABPDS_\varphi$ is $((p,\Epath\Next\psi),wa)$. Then, $\ABPDS_\varphi$ selects one set $X$ (this models the existential quantifier $\Epath$) with $(p,a)\Delta X$ and sends a copy of $\ABPDS_\varphi$ to each of the successor state in $X$ (this models the temporal operator $\Next$). Rule 15 is responsible for simulating the alternating automaton at the propositional level and 16 ensures that the acceptance state of the alternating automaton also yields an accepting path of $\ABPDS_\varphi$. The final states include states of type $P\times\cl_\Release(\varphi)$ to ensure that formulae are accepted which are \emph{never released}. For further details of the standard functioning we refer to~\cite{song2014efficient}. We note that rules 5-14 had to be modified to work with ABPDSs. For a proof sketch of the next Lemma, we refer to the more general Lemma~\ref{lemma:mcheck-CTL-CABPDS}.


\begin{lemma}\label{lemma:mcheck-CTL-ABPDS}
Using the notation above, we have the following: 
$\ABPDS,(p,w),\lab\models\psi$ iff $((p,\psi),w)\in L(\ABPDS_{\varphi})$ for all $\psi \in \cl(\phi)$.
\end{lemma}}

%

\full{\subsection{Compact ABPDS}}
Our reduction of model checking \ral to an acceptance problem over ABPDSs relies on an encoding of an $1$-unbounded {\rbmi} as an ABPDS. Roughly speaking, the stack is used to keep track of the shared pool of resources. A technical difficulty is that an action may consume several resources at a time, whereas an ABPDS can only read the top stack symbol. Therefore, we introduce a more compact encoding of an ABPDS which allows to read (and pop) more than one stack symbol at a time. 

Given  a natural number $r\geq 1$, an \emph{$r$-compact} ABPDS (CABPDS) is a tuple $\CABPDS=(P,\Gamma,\Delta,F,r)$ where all ingredients have the same meaning as in an  ABPDS with the exception that $\Delta \subseteq P\times \Gamma^{\leq r} \times \powerset{P\times\Gamma^*}$ where $\Gamma^{\leq r} = \bigcup_{i=1}^r\Gamma^i$ denotes the set of all non-empty words over $\Gamma$ of length at most $r$. This models that the selection of the next transition can depend on up to the top $r$ stack symbols. All notions introduced so far are also used for CABPDSs. \full{Note that in a configuration $(p,a_1\ldots a_n)$ a transition $(p,b_1\ldots b_j)\Delta\{(p_1,w_1),\ldots,(p_m,w_m)\}$ can only be taken if, and only if, $n\geq j$ and $a_{n-j+1}\ldots a_n=b_1\ldots b_j$. In that case $(p,a_1\ldots a_n)\Rightarrow_\CABPDS\{(p_1,a_1\ldots a_{n-j} w_1),\ldots,\linebreak(p_m,a_1\ldots a_{n-j} w_m)\}$. 
Obviously,  a $1$-compact ABPDS is simply an ABPDS.} \full{\\ }
\full{\indent}A  CABPDS is no more expressive than a ``standard'' ABPDS.\short{ In~\cite{Bulling15ABPDS-xarchive} we show how to construct an ABPDS $\ABPDS(\CABPDS)$ from a given CABDS $\CABPDS$ such that  $c\in L(\CABPDS)$ iff $c\in L(\ABPDS(\CABPDS))$, for all configurations $c\in P\times\Gamma^*$. As for ABPDSs,  we can use CABPDSs as models for \LogCTL.} \full{Essentially, the  top $r$ stack symbols can be encoded in the states of an ABPDS. We make this intuition precise. Let $\CABPDS$ be the $r$-compact ABPDS given above. We define the  ABPDS $\ABPDS(\CABPDS)=(P',\Gamma,\Delta',F)$ consisting of the following elements:
\begin{itemize}
\item $P'=P\cup P\times \Gamma^{\leq r} \times \Gamma^{\leq r-1}$ where states in $P$ are called \emph{real states}, and all other states  \emph{storage states}. 
A storage state has the form $(p,w,v)$ and encodes that word $w$ should be popped from the stack where its prefix $v$ remains to be popped.
\item $\Delta'\subseteq P'\times\Gamma\times\powerset{P'\times\Gamma^*}$ is the smallest relation satisfying the following properties:
For all $p \in P, a \in \Gamma, w,v \in \Gamma^+$ and $X\subseteq P\times\Gamma^*$:
\begin{enumerate}
\item If $(p,a)\Delta X$ then $(p,a)\Delta' X$;
\item If $(p,wa)\Delta X$ then $(p,a)\Delta'((p,wa,w),\epsilon)$;
\item  $((p,vaw,va),a)\Delta' ((p,vaw,v),\epsilon)$ provided that $(p,vaw)\Delta X$;
\item  $((p,aw,a),a)\Delta'X$ provided that $(p,aw)\Delta X$.
\end{enumerate}
\end{itemize}
We briefly explain these  conditions. When a transition of $\CABPDS$  pops only a single symbol from the stack, 
it is also a transition in $\ABPDS(\CABPDS)$ (rule 1). If a transition of $\CABPDS$  pops more than one
symbol the transition is split into several in $\ABPDS(\CABPDS)$. The first symbol and the transition to a storage state is described by Condition 2: the top symbol $a$ is popped from the stack and the next state is $(p,wa,w)$. It expresses that $wa$ should be popped and that $w$ remains to be popped ($a$ has been popped by this very rule); note that $\epsilon$ in this transition means no symbol is pushed on the stack. Condition 3 describes how to pop a single symbol $a$ from a storage node where  $va$ is the word which remains to be popped in order to  complete the simulation of the transition $(p,vaw)\Delta{}X$. Finally, the last rule is applied if all but the last symbols $a$ of the  transition $(p,aw)\Delta{}X$ that is being simulated is read and can be popped from the stack. This completes the encoding and we  obtain  the following result:
\begin{proposition}\label{prop:ABPDS-and-CABPDS}
For any $r$-compact ABPDS $\CABPDS=(P,\Gamma,\Delta,F,r)$  we have that $c\in L(\CABPDS)$ iff $c\in L(\ABPDS(\CABPDS))$, for all configurations $c\in P\times\Gamma^*$. 
\end{proposition}
\full{
\begin{proof}
In the following, we define a function $f$ which translates a run $\rho \in \Runs_\CABPDS$ into one of $\Runs_{\ABPDS(\CABPDS)}$.
Let $\rho = (p,w)(\rho_1,\ldots,\rho_n)$\footnote{We denote with $c(\rho_1,\ldots,\rho_n)$
a tree with root $c$ which has $n$ direct sub-tree $\rho_1,\ldots,\rho_n$ starting at the child nodes of $c$.} and 
$(p_i,w_i)$ be the root of $\rho_i$. This means $(p,w) \Rightarrow_\CABPDS \{(p_i,w_i)\mid 1 \leq i \leq n\}$.
Then, $f$ is defined
by induction on the structure of $\rho$ as follows:
\begin{itemize}
\item if $(p,w) \Rightarrow_\CABPDS \{(p_i,w_i)\mid 1 \leq i \leq n\}$ is generated by $(p,a) \Delta\{(p_i,w_i')\mid 1 \leq i \leq n\}$
for some $a\in \Gamma$, then $f(\rho) = (p,w)(f(\rho_1),\ldots,f(\rho_n))$; and
\item if $(p,w) \Rightarrow_\CABPDS \{(p_i,w_i)\mid 1 \leq i \leq n\}$ is generated by $(p,a_1\ldots a_m) \Delta\{(p_i,w_i')\mid 1 \leq i \leq n\}$
for some $m > 1$ and $a_1,\ldots,a_m \in \Gamma$ (i.e., $w = w'a_1\ldots a_m$), then 
$f(\rho) = (p,w)(((p,a_1\ldots a_m,a_1\ldots,a_{m-1}),w'a_1\ldots a_{m-1})(\ldots(((p,a_1\ldots a_m,\linebreak a_1),w'a_1)(f(\rho_1),\ldots,f(\rho_n)))\ldots))$.
\end{itemize} 
Furthermore, for any path $\kappa \in f(\rho)$, we can pinpoint, by abuse of notation, the corresponding path $f^{-1}(\kappa)$ in $\rho$ as
\begin{itemize}
\item $f^{-1}((p,w)(p_i,w_i)\ldots) = (p,w)f^{-1}((p_i,w_i)\ldots)$; and
\item $f^{-1}((p,w)((p,a_1\ldots a_m,a_1\ldots a_{m-1}),w'a_1\ldots a_{m-1})\ldots((p,a_1\ldots a_m,a_1),w'\linebreak a_1)(p_i,w_i)\ldots) = 
(p,w)f^{-1}((p_i,w_i)\ldots)$.
\end{itemize}
Obviously, any state occurring infinitely often in $f^{-1}(\kappa)$ also appears infinitely often in $\kappa$.

\noindent$(\Rightarrow):$ Assume that $c \in L(\CABPDS)$, then there exists an accepting run $\rho \in \Runs_\CABPDS(c)$.
Then $f(\rho)\in \Runs_{\ABPDS(\CABPDS)}(c)$. For each $\kappa \in f(\rho)$, $f^{-1}(\kappa)$ is accepting, i.e.,
some state in $F$ occurs infinitely often in $\kappa$; hence, it also occurs infinitely often in $\kappa$, showing that
$f(\rho)$ is accepting, i.e., $c \in L(\ABPDS(\CABPDS))$.

\noindent$(\Leftarrow):$ The proof is done analogously to the above case where the function $f^{-1}$ is used to translate
an accepting run $\rho \in \Runs_{\ABPDS(\CABPDS)}(c)$ into that of $\Runs_\CABPDS(c)$.\qed
\end{proof}}
The next corollary is easy to see: the language $L(\CABPDS)=L(\ABPDS(\CABPDS))$ is regular as  $L(\ABPDS(\CABPDS))$ and $P \times \Gamma^*$ are regular and regular languages are closed under intersection.
\begin{corollary}
$L(\CABPDS)$ is regular. 
\end{corollary}
\vspace{-0.4cm}


\subsection{CTL Model checking over compact ABPDSs}

In this section, we consider  model checking \CTL over CABPDSs. Given an $r$-compact ABPDS $\CABPDS=(P,\Gamma,\Delta,F,r)$, a regular labelling function 
$\lab : \Pi \to 2^{P\times\Gamma^*}$ and a \CTL-formula $\varphi$ over $\Props$.
Assume that for each $a \in \Pi^+(\phi)$, $\lab(a)$ is accepted by an alternating automaton
$\AAut_\prop{p} = (S_\prop{p},\Gamma,\delta_\prop{p}, I_\prop{p},F_\prop{p})$; and for each $\prop{p} \in \Pi^-(\phi)$, 
the complement $P \times \Gamma^* \setminus \lab(\prop{p})$ is accepted by an alternating automaton
$\AAut_{\neg \prop{p}} = (S_{\neg \prop{p}},\Gamma,\delta_{\neg \prop{p}}, I_{\neg \prop{p}},F_{\neg \prop{p}})$. We assume the same notational conventions wrt. disjointness and renaming as discussed in Section~\ref{sec:CTLABPDS}. We define the $r$-compact ABPDS $\CABPDS_{\phi}=(P',\Gamma,\Delta',F',r)$ as follows:
\begin{itemize}
\item $P' = (P \times \cl(\phi)) \cup 
            \bigcup_{\prop{p} \in \Pi^+(\phi)} S_{\prop{p}} \cup 
			\bigcup_{{\prop{p}} \in \Pi^-(\phi)} S_{\neg \prop{p}}$;
\item $F' = (P \times \cl_\Release(\phi)) \cup \bigcup_{\prop{p} \in \Pi^+(\phi)} F_{ \prop{p}} \cup 
			\bigcup_{{\prop{p}} \in \Pi^-(\phi)} F_{\neg \prop{p}}$; and
\item $\Delta'$ is the smallest relation satisfying rules 1-14, 16  given in Section~\ref{sec:CTLABPDS} where symbol $a$ is replaced by word $w$ everywhere, and rule 15 of Section~\ref{sec:CTLABPDS} is taken without any changes.
\end{itemize}

The intuition of the construction of 
$\CABPDS_{\phi}$ is the same as for $\ABPDS_\varphi$. We obtain the result:

\begin{lemma}
\label{lemma:mcheck-CTL-CABPDS}
Using the notation above, we have the following: 
$\CABPDS,(p,w),\lab \models \psi$ iff $((p,\psi),w) \in L(\CABPDS_{\phi})$
for all $\psi \in \cl(\phi)$.
\end{lemma}

\full{\begin{proof}[Sketch]
$(\Rightarrow):$ Assume that $\CABPDS,(p,w),\lab \models \psi$. We prove that
$((p,\psi),w)$ has an accepting run in $\CABPDS_{\phi}$ by induction on the
structure of $\psi$.

\inductioncase{Case $\psi = \prop{p}$:} Since $\CABPDS,(p,w),\lab \models \psi$, $(p,w) \in \lab(\prop{p})$.
Then, in $\AAut_\prop{p}$ we have $p_\prop{p} \stackrel{w}{\rightarrow}_{\AAut_\prop{p}}^* S'$ where $p_\prop{p} \in I_\prop{p}$ and  $S' \subseteq F_\prop{p}\subseteq F'$. 
Then, $((p,\prop{p}),w) \Rightarrow_{\CABPDS_\phi} (p_\prop{p},w) \Rightarrow^*_{\CABPDS_\phi} \{ (s',\#) \mid s' \in S'\} \Rightarrow_{\CABPDS_\phi} 
\{ (s',\#) \mid s' \in S'\} \Rightarrow_{\CABPDS_\phi} \ldots$ (by rules 1, 15 and 16) which is an accepting run in $\Runs_{\CABPDS_{\phi}}$.

\inductioncase{Case $\psi = \psi_1\lor\psi_2$:}
Since $\CABPDS,(p,w),\lab \models \psi_1\lor\psi_2$, 
$\CABPDS,(p,w),\lab \models \psi_1$ or 
$\CABPDS,(p,w),\lab \models \psi_2$. 
Without loss of generality, let us assume that $\CABPDS,(p,w),\lab \models \psi_1$.
By induction hypothesis, we have
$((p,\psi_1),w)\in L(\CABPDS_\varphi)$, i.e., there is an accepting $((p,\psi_1),w)$-run $\rho$.
Then, we construct a $((p,\psi_1\lor\psi_2),w)$-run as $((p,\psi_1\lor\psi_2),w)(\rho)$ which
is obviously accepting. Hence, $((p,\psi_1\lor\psi_2),w)\in L(\CABPDS_\varphi)$.

\inductioncase{Case $\psi = \psi_1\land\psi_2$:}
Since $\CABPDS,(p,w),\lab \models \psi_1\land\psi_2$, 
$\CABPDS,(p,w),\lab \models \psi_1$ and
$\CABPDS,(p,w),\lab \linebreak \models \psi_2$. 
By induction hypothesis, we have
$((p,\psi_1),w)\in L(\CABPDS_\varphi)$ and $((p,\psi_2),w)\in L(\CABPDS_\varphi)$, 
i.e., there exist a $((p,\psi_1),w)$-run $\rho_1$
and a $((p,\psi_2),w)$-run $\rho_2$ which are both accepting.
Then, we construct a $((p,\psi_1\land\psi_2),w)$-run as $((p,\psi_1\land\psi_2),w)(\rho_1,\rho_2)$ which
is obviously accepting. Hence, $((p,\psi_1\land\psi_2),w)\in L(\CABPDS_\varphi)$.

\inductioncase{Case $\psi = \Epath\Next\psi_1$:}
Since $\CABPDS,(p,w),\lab \models \psi$, there exists a $(p,w)$-run $\rho = (p,w)(\rho_1,\ldots,\linebreak\rho_n)$\footnote{Recall
that $c(\rho_1,\ldots,\rho_n)$ denotes a tree which is rooted at $c$ and has $n$ direct sub-trees $\rho_1,\ldots,\rho_n$.}
for all roots $(p_i,w_i)$ of $\rho_i$,  
$\CABPDS,(p_i,w_i),$ $\lab \models \psi_1$.
By induction hypothesis, for all $1 \leq i \leq n$, we have $((p_i,\psi_1),w_i) \in L(\CABPDS_{\phi})$, i.e.,
there exists an accepting $((p_i,\psi_1),w_i)$-run $\rho_i'$.
Then, we construct a $((p,\Epath\Next\psi_1),w)$-run as $((p,\Epath\Next\psi_1),w)(\rho_1',\linebreak\ldots,\rho_n')$ which
is obviously accepting. Hence, $((p,\Epath\Next\psi_1),w)\in L(\CABPDS_\varphi)$.

\inductioncase{Case $\psi = \Epath\psi_1\NUntil\psi_2$:}
Since $\CABPDS,(p,w),\lab \models \psi$, there exists a $(p,w)$-run $\rho$ such that,
for all paths $\kappa = (p^\kappa_0,w^\kappa_0) (p^\kappa_1,w^\kappa_1) \ldots \in \rho \in \Runs_\CABPDS((p,w))$, 
$\exists i_\kappa \geq 0$ such that $\CABPDS,(p^\kappa_{i_\kappa},w^\kappa_{i_\kappa}),$ $\lab \models \psi_2$ and
$\forall 0 \leq j < i_\kappa$ we have that $\CABPDS,(p^\kappa_j,w^\kappa_j),\lab \models \psi_1$. 
By induction hypothesis, we have $((p^\kappa_{i_\kappa},\psi_2),w^\kappa_{i_\kappa}) \in L(\CABPDS_{\phi})$ and
$((p^\kappa_j,\psi_1),w^\kappa_j) \in L(\CABPDS_{\phi})$ for all $0 \leq j < i_\kappa$.
Now, we show that $((p^\kappa_i,\psi),w^\kappa_i) \in L(\CABPDS_\phi)$ for all
$\kappa\in\rho$ and $0 \leq i \leq i_\kappa$ by induction on $i_\kappa-i$ (the claim follows for $i=0$, then we have $((p,\psi),w) \in L(\CABPDS_\phi)$):
\begin{description}
\item [Base case: ] Assume that $i = i_\kappa$, then $((p^\kappa_{i_\kappa}),\psi),w^\kappa_{i_\kappa}) \Rightarrow_{\CABPDS_\phi} ((p^\kappa_{i_\kappa},\psi_2),w^\kappa_{i_\kappa})$ by rule 7. Since 
$((p^\kappa_{i_\kappa},\psi_2),w^\kappa_{i_\kappa}) \in L(\CABPDS_{\phi})$, we have that 
$((p^\kappa_{i_\kappa}),\psi),w^\kappa_{i_\kappa})\in L(\CABPDS_{\phi})$.
\item [Induction step: ] 
Assume that $((p^\kappa_i,\psi),w^\kappa_i) \in L(\CABPDS_\phi)$ for all
$\kappa\in\rho$ and $0 < i \leq i_\kappa$. Consider an arbitrary
$\kappa\in\rho$ and $((p^\kappa_{i-1},\psi),w^\kappa_{i-1})$.
Then, there exists a transition $(p^\kappa_{i-1},w^\kappa_{i-1}) \Delta X$ such that
$X = \{ (p^{\kappa'}_{i},w^{\kappa'}_{i}) \mid \kappa' \in \rho, 
(p^{\kappa'}_{i-1},w^{\kappa'}_{i-1})=(p^\kappa_{i-1},w^\kappa_{i-1}) \}$;
i.e., the transition taken at $(p^\kappa_{i-1},w^\kappa_{i-1})$ in $\rho$.
Then, by induction hypothesis, 
$((p^{\kappa'}_{i},\psi),w^{\kappa'}_{i}) \in L(\CABPDS_\phi)$
for all $\kappa' \in \rho$ with
$(p^{\kappa'}_{i-1},w^{\kappa'}_{i-1})=(p^\kappa_{i-1},$ $w^\kappa_{i-1})$;
i.e., $((p',\psi),w') \in L(\CABPDS_\phi)$ for all $(p',w') \in X$.
Moreover, we have that (i)
$((p^\kappa_{i-1},\psi),w^\kappa_{i-1}) \Rightarrow_{\CABPDS_\varphi}
\{ ((p^\kappa_{i-1},\psi_1),w^\kappa_{i-1}) \}
\cup 
\{ ((p',\psi),w') \mid (p',w') \in X \}$ and
(ii) $((p^\kappa_{i-1},\psi_1),w^\kappa_{i-1}) \in L(\CABPDS_\phi)$.
Therefore, $((p^\kappa_{i-1},\psi),w^\kappa_{i-1}) \in L(\CABPDS_\phi)$.
\end{description}
\noindent For the other cases of $\psi$, the proofs are similar.

\noindent$(\Rightarrow):$ Assume that $((p,\psi),w) \in L(\CABPDS_{\phi})$, then we prove that $\CABPDS,(p,w),\lab \models \psi$ by induction on the structure of $\psi$.

\inductioncase{Case $\psi = \prop{p}$:} Since $((p,\prop{p}),w) \in L(\CABPDS_{\phi})$, there exists an accepting 
$(p,\prop{p}),w)$-run $\rho$. Furthermore, the prefix of $\rho$ must  satisfy the following:
$((p,\prop{p}),w) \Rightarrow_{\CABPDS_\varphi} (p_\prop{p},w) \Rightarrow_{\CABPDS_\varphi}^* (f,\#)$ for some $f \in F_\prop{p}$.
Thus, $\AAut_\prop{p}$ has a run $p_\prop{p} \stackrel{w}{\rightarrow}^*_{\AAut_\prop{p}} f$, i.e., $(p,w) \in L(\AAut_\prop{p}) = \lab(\prop{p})$.
Hence, $\CABPDS,(p,w),\lab \models \prop{p}$

\inductioncase{Case $\psi = \psi_1\lor\psi_2$:}
Since $((p,\psi_1\lor\psi_2),w) \in L(\CABPDS_{\phi})$, there exists an accepting 
$((p,\psi_1\lor\psi_2),w)$-run $\rho$. Furthermore, $\rho$ must have the form
$((p,\psi_1\lor\psi_2),w)(\rho')$ where $\rho'$ is rooted at either
$((p,\psi_1),w)$ or $((p,\psi_2),w)$.  
Without loss of generality, let us assume $((p,\psi_1),w)$ is the root of $\rho'$;
then $\rho'$ is an accepting run, i.e., $((p,\psi_1),w) \in L(\CABPDS_{\phi})$
By induction hypothesis, we have $\CABPDS,(p,w),\lab \models \psi_1$; thus,
$\CABPDS,(p,w),\lab \models \psi_1 \lor \psi_2$.

\inductioncase{Case $\psi = \psi_1\land\psi_2$:}
Since $((p,\psi_1\land\psi_2),w) \in L(\CABPDS_{\phi})$, there exists an accepting 
$((p,\psi_1\land\psi_2),w)$-run $\rho$. Furthermore, $\rho$ must have the form
$((p,\psi_1\land\psi_2),w)(\rho_1,\rho_2)$ where $\rho_i$ is rooted at 
$((p,\psi_i),w)$.  
Then, $\rho_1$ and $\rho_2$ are both accepting, i.e., $((p,\psi_1),w) \in L(\CABPDS_{\phi})$
and $((p,\psi_2),w) \in L(\CABPDS_{\phi})$.
By induction hypothesis, we have $\CABPDS,(p,w),\lab \models \psi_1$ and
$\CABPDS,(p,w),\lab \models \psi_2$; thus,
$\CABPDS,(p,w),\lab \models \psi_1 \land \psi_2$.

\inductioncase{Case $\psi = \Epath\Next\psi_1$:}
Since $((p,\Epath\Next\psi_1),w) \in L(\CABPDS_{\phi})$, there exists an accepting 
$((p,\Epath\Next\psi_1),\linebreak w)$-run $\rho$. Furthermore, $\rho$ must have the form
$((p,\Epath\Next\psi_1),w)(\rho_1,\ldots,\rho_n)$ for some
$(p,a_1\ldots a_m) \Delta \{(p_1,w_1),\ldots,(p_n,w_n)\}$
where  $w=w'a_1\ldots a_m$ 
and 
$\rho_i$ is rooted at 
$((p_i,\psi_1),w'w_i')$.  
Then, for all $1\leq i \leq n$, $\rho_i$ is accepting, i.e., $((p_i,\psi_1),w'w_i') \in L(\CABPDS_{\phi})$.
By induction hypothesis, we have $\CABPDS,(p_i,w'w_i'),\lab \models \psi_1$; thus,
$\CABPDS,(p,w),\lab \models \Epath\Next\psi_1$.

\inductioncase{Case $\psi = \Epath\psi_1\NUntil\psi_2$:}
Since $((p,\Epath\psi_1\NUntil\psi_2),w) \in L(\CABPDS_{\phi})$, there exists an accepting 
$((p,\Epath\psi_1\NUntil\psi_2),w)$-run $\rho$. We convert $\rho$ into a prefix $g(\rho)$ of some run in $\CABPDS$
by induction on the structure of $\rho$ as follows:
\begin{itemize}
\item $g(((p,\Epath\psi_1\NUntil\psi_2),w)(\rho')) = (p,w)$ where $\rho'$ is the only direct sub-tree of the root of $\rho$ according to Rule 7 given in Section~\ref{sec:CTLABPDS}; and
\item $g(((p,\Epath\psi_1\NUntil\psi_2),w)(\rho',\rho_1,\ldots,\rho_n)) = (p,w)(g(\rho_1),\ldots,g(\rho_n))$
for some $((p,\linebreak\Epath\psi_1\NUntil\psi_2),w)\Rightarrow_{\CABPDS_\varphi} \{((p,\psi_1),w)\} \cup 
\{((p,\Epath\psi_1\NUntil\psi_2),w') \mid (p',w') \in X\}$ where $(p,w)\Rightarrow_\CABPDS X$ according to Rule 8 given in Section~\ref{sec:CTLABPDS}.
\end{itemize}
Then, every path $\kappa = (p_0,w_0)\ldots(p_m,w_m) \in g(\rho)$ (for some $m \geq 0$, $p_0 = p$, and $w_0=w$) corresponds to a prefix of a path in $\rho$ 
which has the form
$((p_0,\Epath\psi_1\NUntil\psi_2),w_0)\linebreak\ldots((p_m,\Epath\psi_1\NUntil\psi_2),w_m)((p_m,\psi_2),w_m)$ for some
$m \geq 0$. Furthermore, for all $i < m$, $((p_i,\psi_1),w_i)$ is the direct child of 
$((p_i,\Epath\psi_1\NUntil\psi_2),w_i)$. 
Then, for all $i < m$, $((p_i,\psi_1),\linebreak w_i) \in L(\CABPDS_{\phi})$ and $((p_m,\psi_2),w_m) \in L(\CABPDS_{\phi})$.
By induction hypothesis, we have $\CABPDS,(p_i,w_i\linebreak),\lab \models \psi_1$ for all $i < m$ and
$\CABPDS,(p_m,w_m),\lab \models \psi_2$; thus
$\CABPDS,(p,w),\lab \models \Epath\psi_1\NUntil \psi_2$.

\noindent \\
For the other cases of $\psi$, the proofs are similar.
\qed

\end{proof}}
The following theorem follows from this Lemma~\ref{lemma:mcheck-CTL-CABPDS}, Proposition~\ref{prop:ABPDS-and-CABPDS}, and  Theorem~\ref{theo:song1}.}
\short{We obtain the following result:}
\begin{theorem}\label{theorem:CTL-CABPDS}
For a given CABPDS \CABPDS, a regular labelling function $\lab$, and a \CTL-formula $\varphi$ there is an effectively computable alternating automaton $\AAut_{\CABPDS,\varphi}$  such that for all configurations $c=(p,w)\in\Cnf_\CABPDS$ the following holds: $\CABPDS,c,\lab\models\varphi$ iff $((p,\varphi),w)\in L(\AAut_{\CABPDS,\varphi})$.

\end{theorem}
\full{
\begin{proof}
Let \CABPDS be an CABPDS, $\varphi$ a \LogCTL-formula, $(p,w)$ a configuration, and $\lab$ a regular labelling function. We construct the CABPDS  $\CABPDS_\varphi$ with $(\star)$\ $\CABPDS,(p,w)\models\varphi$ iff $((p,\varphi),w)\in L(\CABPDS_\varphi)$ according to Lemma~\ref{lemma:mcheck-CTL-CABPDS}. We apply Proposition~\ref{prop:ABPDS-and-CABPDS} to obtain: $(\star)$ iff $((p,\varphi),w)\in L(\ABPDS(\CABPDS_\varphi))$ where $\ABPDS(\CABPDS_\varphi)$ is an ABPDS. Finally, by Theorem~\ref{theo:song1} we can conclude that there is an effectively constructable alternating $\ABPDS(\CABPDS_\varphi)$-automaton $\AAut$ with $(\star)$ iff $((p,\varphi),w)\in L(\AAut)$.\qed
\end{proof}
}

\full{\section{Decidability of \ral \full{over 1-Unbounded Models}\short{with a Single Unbounded Resource}}}
\short{\section{(Un-)Decidable Model Checking Result}\label{sec:undecidable}}\label{sec:decidable}

\newcommand{\encone}[1]{\ensuremath{\textbf{[}#1\textbf{]}_\textbf{1}}}
\newcommand{\encten}[1]{\ensuremath{\textbf{[}#1\textbf{]}_\textbf{10}}}
\newcommand{\mn}{\ensuremath{\mathsf{\Delta{}max}}}
\newcommand{\con}{\ensuremath{\mathsf{\Delta{}con}}}
\newcommand{\prd}{\ensuremath{\mathsf{\Delta{}prd}}}

\short{
First, we consider the general case of $k$-unbounded \rbmi{}s\footnote{Note that undecidability proofs wrt. \rbmi{}s are stronger than those for \rbm{}s.} with $k\geq 2$. In~\cite{Bulling15ral-IJCAI,Bulling/Farwer:10a} it is shown that most variants of \ral with two resource types are undecidable. This has been proved by reductions of the halting problem of two-counter automata~\cite{HopcroftUllmann79Automata} to the different model checking problems. Two counter automata are finite automata extended with two counters. The undecidability proofs of~\cite{Bulling15ral-IJCAI,Bulling/Farwer:10a} can be adapted to our setting.
We obtain the following result:

\begin{proposition}[Corollary of~\cite{Bulling/Farwer:10a,Bulling15ral-IJCAI}]Model checking \ral (with shared resources) over $k$-unbounded ${\rbmi}s$ with $k\geq 2$ is undecidable.
\end{proposition}
}

\short{Second, we show decidability over $1$-unbounded \rbm{}s. }\full{Throughout this section }\short{In the following }we assume that $\model =
(\Agt,Q,\Pi,\pi,\Act,d,o,\short{\Resources},\full{\mathfrak{R},}t)$ \short{ consists of }\full{is an $1$-unbounded \rbm where $\mathfrak{R}$ is a shared resource structure consisting} of a single unbounded shared resource. Moreover, let  $A$ be a set of agents and $\bar{A}=\Agt\backslash A$. As there is only one resource, we can simplify the notation. We write $\eta$ for $\eta(r)$, $\consumption(\alpha)$ instead of $\consumption(\alpha,r)$ and so on. Also, for an action profile $\vec{\alpha}_A$ we use $\consumption(\vec{\alpha}_A)$ (resp. $\production(\vec{\alpha}_A)$) to refer to $\sum_{a\in A}\consumption(\alpha_a)$ (resp. $\sum_{a\in A}\production(\alpha_a)$). Furthermore, for a natural number $x$, $\encone{x}$ is used to refer to a sequence $||\ldots|$ of $x$ lines each representing one element on the stack, i.e. $\encone{x}$ corresponds to the \emph{unary encoding} of $x$. We write  $\encone{0}=\epsilon$. Similarly, we use $\encten{y}$ to refer to the \emph{ternary encoding} of $y=\encone{x}$ for  a natural number $x$.
\full{\subsection{Encoding of an \rbmi}}
We define the following auxiliary functions where  $q$ is a state in $\model$, $\vec{\alpha}_A$ a joint action of $A$ and $\vec{\alpha}_{\bar{A}}$ a joint action of $\bar{A}$:
\short{
$\mn_{\bar A}(q) = \max\{\consumption(\vec{\alpha}_{\bar{A}}) \mid  \vec{\alpha}_{\bar{A}} \in d_{\bar A}(q)\}$;
 $\con_A(q,\vec{\alpha}_A) = \consumption(\vec{\alpha}_A)+\mn_{\bar A}(q)$; and  
$\prd_A(q,\vec{\alpha}_A,\vec{\alpha}_{\bar{A}}) = \mn_{\bar A}(q)-\consumption(\vec{\alpha}_{\bar{A}})+\production((\vec{\alpha}_A,\vec{\alpha}_{\bar{A}}))$.
}
\full{\begin{eqnarray*}
\mn_{\bar A}(q) &=& \max\{\consumption(\vec{\alpha}_{\bar{A}}) \mid  \vec{\alpha}_{\bar{A}} \in d_{\bar A}(q)\}\\
\con_A(q,\vec{\alpha}_A) &=& \consumption(\vec{\alpha}_A)+\mn_{\bar A}(q)\\
\prd_A(q,\vec{\alpha}_A,\vec{\alpha}_{\bar{A}}) &=& \mn_{\bar A}(q)-\consumption(\vec{\alpha}_{\bar{A}})+\production((\vec{\alpha}_A,\vec{\alpha}_{\bar{A}}))
\end{eqnarray*}}
The number $\mn_{\bar A}(q)$ denotes the worst case consumption of resources of the opponents at $q$, that is the maximal amount of resources they could claim.  The number $\con_A(q,\vec{\alpha}_A)$ is the consumption of resources if $A$ executes $\vec{\alpha}_A$ and the opponents choose their actions with the worst case consumption; this models a pessimistic view. This is valid as the proponents can never be sure to have more resources available. Finally,  $\prd_A(q,\vec{\alpha}_A,\vec{\alpha}_{\bar{A}})$ denotes the number of resources that need to be produced after $(\vec{\alpha}_A,\vec{\alpha}_{\bar{A}})$ was executed at $q$. It is the sum of the number of resources produced by $(\vec{\alpha}_A,\vec{\alpha}_{\bar{A}})$,  and the difference between the consumption of the estimated worst case behavior of the opponents and the consumption of the actions which were actually executed by $\bar{A}$. \full{We state the following  lemma which is fundamental for the correctness of the encoding defined below. It justifies that we can first assume the worst-case behavior of the opponents before correcting this choice.

\begin{lemma}\label{lemma:prod}
Let $\vec{\alpha}=(\vec{\alpha}_A,\vec{\alpha}_{\bar{A}})$ be a tuple consisting of an action profile $\vec{\alpha}_A$ of $A$,  $\vec{\alpha}_{\bar{A}}$ be one of $\bar{A}$ and  $q$ be a state in $\model$. We have that 
\begin{itemize}
\item[(a)] $\production(\vec{\alpha})-\consumption(\vec{\alpha})=\prd_A(q,\vec{\alpha}_A,\vec{\alpha}_{\bar{A}})-\con_A(q,\vec{\alpha}_A)$; and
\item[(b)]the following statements are equivalent for any natural number $x$: 
\begin{itemize}
\item[(i)] for all $\vec{\alpha}'\in \Act_{\bar{A}}$: $x \geq \sum_{a \in \bar{A}}\consumption(\alpha'_a)+\sum_{a \in A}\consumption(\alpha_a)$.
\item[(ii)] $x\geq  \con_A(q,\vec{\alpha}_A)$.
\end{itemize}
\end{itemize}
\end{lemma}
\begin{proof}
\begin{itemize}
\item[(a)] $\prd_A(q,\vec{\alpha}_A,\vec{\alpha}_{\bar{A}})-\con_A(q,\vec{\alpha}_A) = 
     \mn_{\bar A}(q)-\consumption(\vec{\alpha}_{\bar{A}})+\production((\vec{\alpha}_A,\linebreak\vec{\alpha}_{\bar{A}})) - 
     (\consumption(\vec{\alpha}_A)+\mn_{\bar A}(q)) = 
      \production((\vec{\alpha}_A,\vec{\alpha}_{\bar{A}})) - (\consumption(\vec{\alpha}_A)+\consumption(\vec{\alpha}_{\bar{A}}))=
      \production(\vec{\alpha})-\consumption(\vec{\alpha})$.
\item[(b)] $\sum_{a \in \bar{A}}\consumption(\alpha'_a)+\sum_{a \in A}\consumption(\alpha_a) \leq x$ for all $\vec{\alpha}'\in \Act_{\bar{A}}$
 iff 
 $\consumption(\vec{\alpha}_A)+\mn_{\bar A}(q)\linebreak \leq x$  iff  $\con_A(q,\vec{\alpha}_A) \leq x$.
\end{itemize}\qed
\end{proof}
}


From $\model$ and $A$, we define an $r$-compact ABPDS where  $r=\encone{\max_{q,\vec{\alpha}_A,\vec{\alpha}_{\bar{A}}}\{\linebreak\con_A(q,\vec{\alpha}_A), \prd_A(q,\vec{\alpha}_A,\vec{\alpha}_{\bar{A}})\}}$ is the maximal number which is ever \emph{consumed} or \emph{produced}.
\begin{definition}[$\CABPDS_{\model,A}$]
The $r$-compact ABPDS $\CABPDS_{\model,A}$ is the CABPDS  $(S,\Gamma,\Delta,F,r)$ where $S=F=Q$, $\Gamma=\{|\}$, and for all $q \in Q$, $\vec{\alpha}_A \in d_A(q)$ we have that 
\full{\[}\short{\\$}(q,\encone{\con_A(q,\vec{\alpha}_A)})\Delta\{(o(q,(\vec{\alpha}_A,\vec{\alpha}_{\bar{A}})),\encone{\prd_A(q,\vec{\alpha}_A,\vec{\alpha}_{\bar{A}})}) \mid \vec{\alpha}_{\bar{A}} \in d_{\bar A}(q)\}.\full{\]}\short{$}
\end{definition}
\full{It is easily seen that $\CABPDS_{\model,A}$ is indeed an $r$-compact ABPDS.  The purpose of}\short {The CABPDS } $\CABPDS_{\model,A}$ \full{is to encode}\short{ encodes } the outcome sets $out(q,s_A,\eta)$ for any state $q$ and strategy $s_A$.  \full{Let $w\in\{|\}^*$ and $\rho\in\Runs_{\CABPDS_{\model,A}}$. We define $h(w)=\encten{w}$ and lift $h$ to configurations $h((p,w))=(p,h(w))$, to finite or infinite sequences $c_0c_1\ldots$ of configurations via $h(c_0c_1\ldots)=h(c_0)h(c_1)\ldots$, and to runs $h(\rho)=\{h(\kappa)\mid\kappa\in\rho\}$. Then, the next result states that runs of $\CABPDS_{\model,A}$ are the outcome sets of $A$. First, observe that  for every strategy $s_A$ we have that there is a run $\rho\in\Runs_{\CABPDS_{\model,A}}$ with $h(\rho)=out(q,s_A,\eta)$. The automaton simply chooses the same actions as specified by the strategy. Similarly, in the reverse case, if the automaton takes a transition  corersponding to an action tuple $\vec{\alpha}_A$ after the finite run $b$, then we define the strategy $s_A$ such that $s_A(h(b))=\vec{\alpha}_A$. We note that here it is important that the strategy is perfect recall and takes the hisotry of  states as well as of shared endowments into account.

\begin{lemma}[Encoding Lemma]\label{lemma:encoding}
$h:\Runs_{\CABPDS_{\model,A}}\rightarrow \{out(q,s_A,\eta)\mid (q,\eta)\in Q\times\Enments \linebreak\text{ and $s_A$ is a strategy of $A$}\}$ is an isomorphism.
\end{lemma}}
\full{\begin{proof}[Sketch]
The proof is done by induction on the number of simulation steps. 
Let $s_A$ be a strategy and $t=out(q,s_A,\eta)$ the $(q,s_A,\eta)$-outcome. First, we argue that there is a run $\rho\in\Runs_{\CABPDS_{\model,A}}((q,\encone{\eta}))$ with $h(\rho)=t$. Let $t^i$ and $\rho^i$ be the finite version of $t$ and $\rho$ up to depth $i\geq 0$, respectively. We construct $\rho$ step-by-step. Clearly, $h(\rho^0)=h(t^0)$. Let $b^i$ be any finite branch in $t^i$ with final configuration $(q',\eta')$ and with successor states $\{(q_1,\eta_1),\ldots,(q_n,\eta_n)\}$ and let the action $\vec{\alpha}_{\bar{A}}^j$ of the opponents be the action which led to $(q_j,\eta_j)$ given that $s_A(t^i)=\vec{\alpha}_A$, i.e. $o(q',(\vec{\alpha}_A,\vec{\alpha}^j_{\bar{A}}))=q_j$, for $1\leq j\leq n$. By definition there is a transition $(q',\encone{\con_A(q',\vec{\alpha}_A)})\Delta\{(q_j,\encone{\prd_A(q',\vec{\alpha}_A,\vec{\alpha}_{\bar{A}})}) \mid \vec{\alpha}_{\bar{A}} \in 1\leq  j\leq n\}$. Let $\kappa^i$ be the finite branch on $\rho^i$ corresponding to $b^i$, by induction we have $h(\kappa^i)=b^i$. The last state on $\kappa^i$ is $c'=(q',\encone{\eta'})$. By Lemma~\ref{lemma:prod}(b) and the fact  that  there is a transition after $b^i$, $\eta'\geq \con_A(q',\vec{\alpha}_A)$. Thus, the transition of the automaton can be taken and by  Lemma~\ref{lemma:prod}(a),  $\{(q_1,\encone{\eta_1}),\ldots,(q_n,\encone{\eta_n})\}$ is a direct successor of $c'$.

For the other direction, let $\rho\in\Runs_{\CABPDS_{\model,A}}(c)$. The proof is  done in a similar way. Let $\kappa^i=(q_0,w_0)\ldots(q_n,w_n)$ be a finite branch in $\rho$ and assume that the automaton  takes as next transition $(q_n,\encone{\con_A(q_n,\vec{\alpha}_A)})\Delta\{(o(q_n,(\vec{\alpha}_A,\vec{\alpha}_{\bar{A}})),\encone{\prd_A(q_n,$ $\vec{\alpha}_A,\vec{\alpha}_{\bar{A}})}) \mid \vec{\alpha}_{\bar{A}} \in d_{\bar A}(q)\}$. Then, we define $s_A(h(\kappa^i))=\vec{\alpha}_A$. To see that $s_A$ is well-defined we observe that  $\rho$ cannot contain two finite banches $b$ and $b'$ which are identical.\qed
%
%
\end{proof}}
\full{\subsection{Model Checking \ral over $1$-Unbounded \rbm{}s is Decidable }}
\short{
To show our main decidable result, we
need to 
extend \rbm{}s with regular labelling functions $\pi : \Pi \to 2^{Q\times\Enments}$ as done in Section~\full{\ref{sec:CTLABPDS}}\short{\ref{sec:CTL}} for PDSs.}\full{
In this section we put the pieces together and show that model checking $\ral$ over $1$-unbounded \rbm{}s is decidable.  Before we do so, we need to extend \rbm{}s with regular labelling functions $\pi : \Pi \to 2^{Q\times\Enments}$ as done in Section~\ref{sec:CTLABPDS} for PDSs. Clearly, the ``sate-based''  labelling function $\pi':\Props\rightarrow \powerset{Q}$ in $\model$ is a special regular labelling function with $\pi(\prop{p})=\{(q,\enment)\mid \enment\in\Enments, q\in\pi'(\prop{p})\}$. From now on, we assume that $\pi$ is regular. Our model checking algorithm builds upon model checking \CTL formulae over CABPDSs as outlined in Theorem~\ref{theorem:CTL-CABPDS}.  The main idea is the following. Suppose} \short{Now, suppose that }we want to model check $\model,q_0,\eta\models\coopdown{A}\phi$  where $\coopdown{A}\phi$ is a flat  formula\full{ and $\Epath\varphi$ is in negation normal form}\full{~\footnote{Note that the release operator cannot occur here.}}. Firstly, we  construct the CABPDS $\CABPDS_{\model,A}$ which accepts the outcome sets of $A$\short{. }\full{ by the Encoding Lemma~\ref{lemma:encoding}.} Let $\lab$ be the labelling function defined as: $(q,\encone{\eta})\in\lab(\prop{p})$ iff $(q,\enment)\in\pi(\prop{p})$. Then, we have that:
$\model,q_0,\eta\models\coopdown{A}\phi$ if, and only if,  $\CABPDS_{\model,A},(q,\encone{\enment}),\lab\models \Epath\varphi$. By Theorem~\ref{theorem:CTL-CABPDS} this can be efficiently solved by constructing an alternating automaton $\AAut_{\CABPDS_{\model,A},\Epath\varphi}$ that accepts $((q,\Epath\varphi),\encone{\eta})$ iff the above equivalence is true. \full{This shows the following result:

\begin{proposition}\label{lemma:decidability}
Let the labelling function in $\model$ be regular and $\coop{A}\varphi$ be a flat $\ral$-formula in negation normal form. Then, we can construct an alternating automaton $\AAut_{\CABPDS_{\model,A},\Epath\varphi}$ such that $((q,\Epath\varphi),\encone{\eta})\in L(\AAut_{\CABPDS_{\model,A},\Epath\varphi})$ if, and only if, $\model,q,\eta\models\coopdown{A}\phi$.
\end{proposition}}

\short{Finally, this procedure can be combined with the standard bottom-up model checking approach used for \CTLs~\cite{Clarke99modelchecking}. }
\full{This proposition can be applied recursively to model check an arbitrary $\ral$-formula $\varphi$, following the standard bottom-up model checking approach used for \CTLs~\cite{Clarke99modelchecking}.} Firstly, the innermost (flat) formulae $\psi$ of $\varphi$ are  considered. \full{By Proposition~\ref{lemma:decidability} we}\short{We} can compute the regular set of configurations at which each of these subformulae $\psi$  hold\short{ and}\full{. Then, we} replace the  subformula \full{$\psi$} by a fresh propositions $\prop{p}_{\psi}$\full{ and }\short{. Then, we }extend the regular labelling of $\model$ such that $\prop{p}_{\psi}$ is assigned the  configurations at which $\psi$ is true (Theorem~\ref{theorem:CTL-CABPDS}). Applied recursively, we obtain:

\begin{theorem}\label{theorem:decidability}
The model-checking problem for \ral (with shared resources) over $1$-unbounded {\rbm}s is decidable.
\end{theorem}
\full{\begin{proof}[Sketch] The proof proceeds by induction on the formula structure.
Suppose we want to model check $\model,q,\enment\models \coop{A}\Sometm\phi$. The other cases are handled analogously. Let  $\xi=\coop{B}\chi$ be any strict subformula of $\phi$. By induction hypothesis and Lemma ~\ref{lemma:decidability} we can construct an alternating automaton $\AAut_{\CABPDS_{\model,B},\Epath\chi}$ that  accepts  exactly those $((q',\Epath\chi),h(\enment'))$ with $\model,q',\eta'\models\xi$. Then, we replace $\xi$ in $\phi$ with a fresh proposition $\prop{p}_\xi$ and extend $\pi$ by defining $\pi(p_\xi)=L(\AAut'_{\CABPDS_{\model,B},\Epath\chi})$ where $\AAut'_{\CABPDS_{\model,B},\Epath\chi}$ is the automaton with $L(\AAut'_{\CABPDS_{\model,B},\Epath\chi})=\{(q,\eta)\mid ((q,\Epath\chi),\encone{\eta})\in L(\AAut_{\CABPDS_{\model,B},\Epath\chi})\}$. We proceed with this procedure until the ``updated'' $\phi$ is completely propositional. Then, we can apply Proposition~\ref{lemma:decidability} to check whether  $\model,q,\enment\models \coop{A}\Sometm\phi$.\qed 
\end{proof}}

\full{\section{General Undecidability Result}\label{sec:undecidable}
In~\cite{Bulling15ral-IJCAI,Bulling/Farwer:10a} it has been shown that most variants of \ral are undecidable. This has been proved by reductions of the halting problem of two-counter automata~\cite{HopcroftUllmann79Automata} to the different model checking problems. Two counter automata are finite automata extended with two counters. Transitions depend on the current state of the automaton and on whether the counters are zero or non-zero. If a transition is taken, the automaton may change its control state and may increment or decrement the counters. The basic idea of the reduction is to encode a two counter automaton as an \rbmi\footnote{We note that we show undecidability over \rbmi{}s. Such undecidability result are stronger than for \rbm{}s as the former  is a special case of the latter.}.  Each of the two counters corresponds to a resource type. Agents' actions are used to simulate the selection of a transition and the incrementation and decrementation of counters. The key difficulty is to encode the \emph{zero test}, i.e. to check whether resources are available. The two counter automaton can check if a counter is zero or not in the transition relation by definition. But, if in the resource bounded model a transition should only be taken if no resources are available, there is nothing which can prevent the agent to take the transition even if it has resources available. Clearly, such an inconsistent behavior would break the simulation. Therefore, a second agent, playing the role of a spoiler, is used to check that such inconsistent transitions result in a ``fail states'' which cannot be used to witness an accepting run of the automaton. Then, it is  shown that the two-counter automaton halts on the empty input iff  $\coopdown{1}\mathbf{F}\prop{halt}$ is true in a model which encodes the transition table of the automaton~\cite{Bulling15ral-IJCAI}. In another result the authors of~\cite{Bulling15ral-IJCAI} also show that undecidability is the case for a single agent only. This is achieved by nesting modalities and letting the agent itself play the role of the spoiler:  the two-counter automaton halts on the empty input iff $\coopdown{1}(\neg\coopdown{1}\Next\prop{err})\NUntil\prop{halt}$. These undecidability proofs can be (directly) adapted to our setting; actually, due to the shared resources the technicalities are even simpler. We note that the undecidability proof does not require the full expressivity of strategies as dedined in this paper. Strategies which only take the history of states into account are sufficient to encode the behavior of a two-counter automaton. This corresponds to the fact that the automaton takes transitions based on the control states and whether the counters are zero or non-zero, but not the actual counter value. We refer to~\cite{Bulling15ral-IJCAI,Bulling/Farwer:10a} for further details about the construction. We obtain the following result:

\begin{corollary}[of~\cite{Bulling/Farwer:10a,Bulling15ral-IJCAI}]Model checking \ral  (with shared resoures) over $k$-unbounded ${\rbmi}s$ with $k\geq 2$ is undecidable, even in the following restricted cases:
\begin{enumerate}
\item In the case of a single agent and a fixed formula of the form $\coopdown{1}(\neg\coopdown{1}\Next\prop{p})\NUntil\prop{q}$.
\item In the case of two agents and a fixed formula of the form $\coopdown{1}\mathbf{F}\prop{p}$.
\end{enumerate}
\end{corollary}}
\vspace{-0.7cm}

\section{Conclusions}\label{sec:concl}
In this paper, we have introduced a variant of resource agent logic \ral~\cite{Bulling/Farwer:10a} with shared resources, which can be consumed and produced. We showed that the model checking problem is undecidable in the presence of at least two unbounded resource types. Our main technical result is a decidability proof of model checking \ral with one shared, unbounded resource type. Otherwise, we impose no restrictions, in particular nested cooperation modalities do not reset the resources available to agents. This property is sometimes called \emph{non-resource flatness}. In order to show decidability, we first
show how \CTL can be model-checked with respect to (compact) alternating B\"{u}chi pushdown systems extending results on model checking \CTL over pushdown and alternating pushdown systems~\cite{song2014efficient,bouajjani1997reachability}. A compact alternating B\"{u}chi pushdown system allows to read and to pop more than one symbol from its stack at a time. It is used for  encoding resource bounded models in order to apply the automata-based model checking algorithm.  

\textbf{Acknowledgement.} We would like to thank Natasha Alechina and Brian Logan for the many discussions on this topic and their valuable comments.

\bibliographystyle{plain}
\bibliography{prima,prima-bibs} 

\begin{thebibliography}{10}

\bibitem{Alechina//:14c}
N.~Alechina, B.~Logan, H.~N. Nguyen, and F.~Raimondi.
\newblock Decidable model-checking for a resource logic with production of
  resources.
\newblock In {\em Proceedings of the 21st European Conference on Artificial
  Intelligence ({ECAI}-2014)}, pages 9--14. ECCAI, IOS Press, 2014.

\bibitem{Alechina//:10a}
N.~Alechina, B.~Logan, H.~N. Nguyen, and A.~Rakib.
\newblock Resource-bounded alternating-time temporal logic.
\newblock In {\em Proceedings of the 9th International Conference on Autonomous
  Agents and Multiagent Systems ({AAMAS} 2010)}, pages 481--488. IFAAMAS, 2010.

\bibitem{Bulling15ral-IJCAI}
Natasha Alechina, Nils Bulling, Brian Logan, and Hoang~Nga Nguyen.
\newblock On the boundary of (un)decidability: Decidable model-checking for a
  fragment of resource agent logic.
\newblock In {\em Proceedings of the Twenty-Fourth International Joint
  Conference on Artificial Intelligence, {IJCAI} 2015, Buenos Aires, Argentina,
  July 25-31, 2015}, pages 1494--1501, 2015.

\bibitem{Alur//:02a}
R.~Alur, T.~Henzinger, and O.~Kupferman.
\newblock Alternating-time temporal logic.
\newblock {\em Journal of the {ACM}}, 49(5):672--713, 2002.

\bibitem{bouajjani1997reachability}
Ahmed Bouajjani, Javier Esparza, and Oded Maler.
\newblock Reachability analysis of pushdown automata: Application to
  model-checking.
\newblock In {\em CONCUR'97: Concurrency Theory}, pages 135--150. Springer,
  1997.

\bibitem{bozzelli2007complexity}
Laura Bozzelli.
\newblock Complexity results on branching-time pushdown model checking.
\newblock {\em Theoretical computer science}, 379(1):286--297, 2007.

\bibitem{BullingFarwer09rtl-clima-post}
N.~Bulling and B.~Farwer.
\newblock Expressing properties of resource-bounded systems: The logics {RBTL}
  and {RBTL$^*$}.
\newblock In {\em Post-Proceedings of {CLIMA '09}}, number 6214 in LNCS 6214,
  pages 22--45, 2010.

\bibitem{Bulling/Farwer:10a}
N.~Bulling and B.~Farwer.
\newblock On the (un-)decidability of model checking resource-bounded agents.
\newblock In {\em Proceedings of the 19th European Conference on Artificial
  Intelligence (ECAI 2010)}, volume 215, pages 567--572. IOS Press, 2010.

\bibitem{Bulling15ABPDS}
Nils Bulling and Hoang~Nga Nguyen.
\newblock Model checking resource bounded systems with shared resources via
  alternating {B}{\"u}chi pushdown systems (to appear).
\newblock In {\em Proc. of the 18th International Conference on Principles and
  Practice of Multi-Agent Systems (PRIMA 2015)}, Bertinoro, Italy, October
  2015.

\bibitem{cachat2002symbolic}
Thierry Cachat.
\newblock Symbolic strategy synthesis for games on pushdown graphs.
\newblock In {\em Automata, Languages and Programming}, pages 704--715.
  Springer, 2002.

\bibitem{Clarke99modelchecking}
E.~Clarke, O.~Grumberg, and D.~Peled.
\newblock {\em Model Checking}.
\newblock MIT Press, 1999.

\bibitem{Clarke81ctl}
E.M. Clarke and E.A. Emerson.
\newblock Design and synthesis of synchronization skeletons using branching
  time temporal logic.
\newblock In {\em Proceedings of Logics of Programs Workshop}, volume 131 of
  {\em Lecture Notes in Computer Science}, pages 52--71, 1981.

\bibitem{DellaMonica//:13a}
D.~{Della Monica}, M.~Napoli, and M.~Parente.
\newblock Model checking coalitional games in shortage resource scenarios.
\newblock In {\em Proceedings of the 4th International Symposium on Games,
  Automata, Logics and Formal Verification (GandALF 2013}, volume 119 of {\em
  EPTCS}, pages 240--255, 2013.

\bibitem{HopcroftUllmann79Automata}
J.E. Hopcroft and J.D. Ullman.
\newblock {\em Introduction to Automata Theory, Languages, and Computation}.
\newblock Addison-Wesley, 1979.

\bibitem{song2014efficient}
Fu~Song and Tayssir Touili.
\newblock Efficient {CTL} model-checking for pushdown systems.
\newblock {\em Theoretical Computer Science}, 549:127--145, 2014.

\bibitem{suwimonteerabuth2006efficient-tr}
Dejvuth Suwimonteerabuth, Stefan Schwoon, and Javier Esparza.
\newblock Efficient algorithms for alternating pushdown systems: Application to
  certificate chain discovery with threshold subjects.
\newblock 2006.

\end{thebibliography}
\vspace{-0.5cm}

\end{document}